\documentclass[reqno,12pt]{amsart}
\usepackage[margin=1in]{geometry}
\usepackage{amsmath}
\usepackage{amssymb}
\usepackage{amsthm}
\usepackage{paralist}
\usepackage{framed}
\usepackage{tikz}
\usepackage{hyperref}

\DeclareMathOperator{\diam}{diam}
\DeclareMathOperator{\dist}{dist}
\DeclareMathOperator{\disc}{disc}
\DeclareMathOperator{\area}{area}
\DeclareMathOperator{\arcsinh}{arcsinh}
\DeclareMathOperator{\real}{Re}
\DeclareMathOperator{\imag}{Im}
\DeclareMathOperator{\lead}{lead}
\DeclareMathOperator{\Sturm}{\texttt{Sturm}}
\DeclareMathOperator{\Descartes}{\texttt{Descartes}}
\DeclareMathOperator{\CSturm}{\texttt{CSturm}}
\DeclareMathOperator{\SqFreeEVAL}{\texttt{SqFreeEVAL}}
\DeclareMathOperator{\SqFreeCEVAL}{\texttt{SqFreeCEVAL}}
\DeclareMathOperator{\CEVAL}{\texttt{CEVAL}}
\DeclareMathOperator{\true}{\textsc{True}}
\DeclareMathOperator{\false}{\textsc{False}}
\DeclareMathOperator{\Vol}{Vol}

\newcommand{\eqlabel}{\tag{\theequation a,\theequation b}}

\theoremstyle{definition}
\newtheorem{proposition}{Proposition}[section]
\newtheorem{lemma}[proposition]{Lemma}
\newtheorem{corollary}[proposition]{Corollary}
\newtheorem{remark}[proposition]{Remark}
\newtheorem{example}[proposition]{Example}

\newenvironment{algorithm}[1]
{\refstepcounter{proposition} 
\setlength{\FrameSep}{0.1in}
\MakeFramed{\setlength{\hsize}{4.7in} \FrameRestore}
\centerline{\bf Algorithm \thesection.\arabic{proposition}. #1\vspace{.05in}}}
{\endMakeFramed}

\begin{document}
\title{Applications of Continuous Amortization to Bisection-based Root Isolation}

\author{Michael A. Burr}
\thanks{This work was partially supported by a grant from the Simons Foundation (\#282399 to Michael Burr).}
\address{Department of Mathematical Sciences, Clemson University, Clemson, SC 29634}
\email{burr2@clemson.edu}

\begin{abstract}
Continuous amortization is a technique for computing the complexity of algorithms, and it was first presented by the author in \cite{BurrKrahmerYap:ContinuousAmortization}.  Continuous amortization can result in simpler and more straight-forward complexity analyses, and it was used in \cite{BurrKrahmerYap:ContinuousAmortization,BurrKrahmer:SqFreeEVAL,SharmaYap:ContinuousAmortization} to provide complexity bounds for simple root isolation algorithms.  This paper greatly extends the reach of continuous amortization to serve as an overarching technique which can be used to compute complexity of many root isolation techniques in a straight-forward manner.  Additionally, the technique of continuous amortization is extended to higher dimensions and to the computation of the bit-complexity of algorithms.  In this paper, six continuous amortization calculations are performed to compute complexity bounds (on either the size of the subdivision tree or the bit complexity) for several algorithms (including algorithms based on Sturm sequences, Descartes' rule of signs, and polynomial evaluation); in each case, continuous amortization achieves an {\em optimal} complexity bound.

\end{abstract}

\maketitle

\section{Introduction}
The technique of continuous amortization was first introduced by the author in \cite{BurrKrahmerYap:ContinuousAmortization} and was used in \cite{BurrKrahmer:SqFreeEVAL} and \cite{SharmaYap:ContinuousAmortization} to compute the size of the subdivision tree for two evaluation-based algorithms.  These algorithms are called \texttt{SqFreeEVAL} and \texttt{EVAL}, and they are simple, numerical, real root isolation algorithms whose primitives are based upon the evaluation of a given polynomial and its derivatives at dyadic points.  Some of the advantages of using continuous amortization in these analyses include that the resulting bounds are optimal (or nearly optimal), and the analyses are much simpler than competing methods, e.g., \cite{Yakoubsohn:bisection,SagraloffYap:CEVAL}.

Continuous amortization is an analysis technique that can be applied to bisection-based algorithms over a domain $D$ in $\mathbb{R}^n$.  We recall that a bisection-based algorithm is one which adaptively subdivides $D$ until the resulting subdomains are small enough so that a (usually simple) terminal condition applies.  To use the technique of continuous amortization, a nonnegative stopping function $F:D\rightarrow\mathbb{R}$ must be constructed, such that the value of $F(\vec{x})$, at a fixed $\vec{x}\in D$, is a lower bound on the size of a subregion of $D$ containing $\vec{x}$ for which the terminal condition fails.  Then, the number of subdivisions performed by the subdivision-based algorithm is $O\!\left(\int d\vec{x}/F(\vec{x})\right)$ (more precise details are given in Section \ref{sec:continuousamortization} and generalizations are given in Sections \ref{sec:complex} and \ref{sec:bit}).

This technique is called continuous amortization because the function $1/F$ can be thought of as a charging function for $D$.  In particular, at points $\vec{x}$ where the algorithm may need to perform a large number of subdivisions in order to satisfy a terminal condition, it follows that the leaf of the subdivision tree whose associated subregion includes $\vec{x}$ may be a very deep node and the value of $1/F(\vec{x})$ may be quite large.  More precisely, the value of $\log_2(\diam(D))-\log_2(F(\vec{x}))$ is an upper bound on depth of the subdivision tree at $\vec{x}$, where $\diam(D)$ is the diameter of the region $D$.  Therefore, the value of $1/F$ at points in $D$ is a charge according to how much work may need to be done near that point, and the integral sums these local costs over the entire domain $D$.

In this paper, we apply the technique of continuous amortization to standard root isolation algorithms including Sturm sequences \cite{CollinsLoos:Sturm}, Descartes' rule of signs \cite{CollinsAkritas:Descartes}, and \texttt{CEVAL} \cite{SagraloffYap:8Point}.  The goal of this paper is to provide a collection of examples which show how continuous amortization can be applied to a significant variety of root isolation algorithms.  In particular, in Section \ref{sec:real}, we apply the technique to real root isolation algorithms and achieve very simple proofs for the computation of the sizes of their subdivision trees.  In Section \ref{sec:complex}, we show how to extend the technique of continuous amortization to two dimensional domains and provide more straight-forward proofs for the computation of the sizes of the subdivision trees for complex root isolation.  Finally, in Section \ref{sec:bit}, we once again extend continuous amortization to compute, not the size of the subdivision tree, but the bit-complexity of a few algorithms.

We stress the goals of this paper are \begin{inparaenum}[(1)]
\item to provide a unifying framework for the analysis of the complexity of root isolation techniques;
\item to expand the collection of examples to which continuous amortization applies;
\item to extend the formula for continuous amortization in two ways: to compute in higher dimensions and to compute values other than the size of the subdivision tree; and
\item to provide simpler proofs than those that appear in the literature for the complexity of standard root isolation techniques.
\end{inparaenum}
We point out that the goal of this paper is {\em not} to improve the known bounds (although improved bounds are provided in some cases), but to exhibit applications of continuous amortization.

\section{Background on Subdivision-based Root Isolation}\label{sec:background}

The collection of literature describing root isolation techniques is much too vast to describe here.  Descriptions of the progress of real root isolation can be found in the surveys \cite{Pan:Optimal,Pan:History,McNamee:Numerical,McNameePan:Efficient}.  In this paper, we focus on the algorithms and papers which have the most direct impact on and relationship with the current paper.  This previous work can be separated into three broad categories: Sturm sequence-based techniques, Descartes' rule of signs-based techniques, and evaluation-based techniques (also called Bolzano's theorem-based techniques).

\subsection{Sturm sequence methods}
An algorithm based on Sturm sequences for isolating the real roots of a polynomial was first presented in \cite{CollinsLoos:Sturm}.  This algorithm's predicate determines the precise number of roots in any interval; because the exact number of roots can be counted, the size of the subdivision tree is optimal.  Some of the main details of the algorithm are given in Section \ref{sec:real:sturm}.  The Sturm algorithm is theoretically interesting, but, in practice, even though the Sturm predicate is very strong, other algorithms are preferred because the preprocessing step for the Sturm algorithm can be prohibitively time consuming, e.g., see \cite{HemmerTsigaridasZafeirakopoulosEmirisKaravelasMourrainCrossBenchmarking, EmirisHemmerKaravelasLimbachMourrainTsigaridasZafeiakopoulosSmall}.

It is standard to judge the complexity of root isolation using the benchmark problem of isolating all of the real roots of a degree $d$ polynomial with integer coefficients of bit-size at most $L$.  In \cite{davenport:85}, it was shown that the complexity of Sturm's subdivision tree is $O(d(L+\ln d))$.  In \cite{eigenwillig-sharma-yap:descartes:06}, an example was provided which shows that this bound is tight in the case where $L\geq\ln d$.  The bit complexity of this algorithm was shown to be $\widetilde{O}(d^4L^2)$ in \cite{davenport:85}, where the $\widetilde{O}$ means that logarithmic factors have been suppressed.  The logarithmic part of this bound was improved in \cite{du-sharma-yap:sturm:07}.  For more of the history and alternate proofs, see \cite{CollinsLoos:Sturm,Johnson:Thesis,reischert:subresultant:97,johnson:root-isolation:98,lickteig-roy:sequences:01,Emiris:Complexity}.

Sturm sequence methods can also be adapted to complex root isolation, as in \cite{Pinkert:1976:EMF:355705.355710,Wilf:1978:GBA:322077.322084,CamargoBrunetto200095}.  This algorithm's predicate can determine the number of roots in any bounded rectangle in the complex plane.  Since the predicate gives the exact number of roots, the size of the subdivision tree is optimal.  Some details of this algorithm are given in  and \ref{sec:ComplexSturm}.  This algorithm is rarely used in practice, however, because it requires the Sturm sequence to be recomputed for every query.  In \cite{du-sharma-yap:sturm:07}, it is shown that the size of the complex Sturm subdivision tree is $O(d(L+\ln d))$ for the benchmark problem of isolating all complex roots of a polynomial of degree $d$ and integer coefficients of bit-size at most $L$.  There, it is also shown that the bit complexity may be $\widetilde{O}(d^5L^3)$.

\subsection{Descartes' rule of signs methods}
An algorithm based on Descartes' rule of signs for isolating the real roots of a polynomial was first presented in \cite{CollinsAkritas:Descartes} using the standard power basis.  The algorithm was also described using the Bernstein basis in \cite{LaneRiesenfeldDescartes}, see also \cite{MourrainRouillierRoyDescartes,Mourrain2002612,BasuPollackRoyRealGeometry}.  This algorithm's predicate determines an upper bound on the number of roots in an interval; because the predicate only provides an upper bound, the subdivision tree may be larger than the optimal tree.  Some of the main details of the algorithm can be found in Section \ref{sec:Descartes}.  In practice, however, Descartes' rule of signs seems to be efficient and practical, see, for example \cite{johnson:root-isolation:98,collins-johnson-krandick:cad:02,Rouillier200433,MourrainRouillierRoyDescartes}.  In particular, a history of improvements of this algorithm up to 2004 can be found in  \cite{Rouillier200433}, and a recent practical improvement can be found in \cite{SagraloffComplexity,SagraloffDescartes}.

For the benchmark problem of isolating all of the real roots of a degree $d$ polynomial with integer coefficients of bit-size at most $L$, it is shown in \cite{eigenwillig-sharma-yap:descartes:06} that the subdivision tree is $O(d(L+\ln d))$.  As in the case for Sturm sequences, this bound is optimal when $L\geq \ln d$.  In addition, the authors show in \cite{eigenwillig-sharma-yap:descartes:06} that the bit complexity is $\widetilde{O}(d^4L^2)$.  For other complexity results, see \cite{johnson:root-isolation:98,MR1383367}

\subsection{Evaluation-based methods}
There are a wide variety of evaluation-based approaches to root isolation, e.g., see \cite{moore:bk,mitchell:robust-ray:90,Henrici:search:70,Yakoubsohn:bisection,BurrKrahmerYap:ContinuousAmortization,BurrKrahmer:SqFreeEVAL,SagraloffYap:8Point,SagraloffYap:CEVAL,Chee:Analytic}.  The predicates in these algorithms all involve evaluating a function and its derivatives at points in a domain.  The predicates in these algorithms are typically fairly weak, but, in most cases, they are very simple to implement and can be evaluated efficiently.  Because the individual predicates are so simple, there is hope that algorithms based on these predicates would be efficient in practice.  These methods are also interesting because, unlike more symbolic techniques, evaluation-based techniques can be applied to small domains, they can be generalized to analytic functions \cite{Yakoubsohn:bisection,Chee:Analytic}, and they can be generalized to higher dimensions  \cite{marchingcube1987,snyder:interval:92,plantinga-vegter:isotopic:04,plantinga:thesis:06,linyap2009,burr+3:subdiv2:10,galehouse:thesis}.

The particular algorithms studied in this paper are the $\SqFreeEVAL$ algorithm and the $\SqFreeCEVAL$ algorithm.  $\SqFreeEVAL$ is based on an algorithm of \cite{mitchell:robust-ray:90}, which is, in turn, based on an algorithm of \cite{moore:bk}.  In \cite{BurrKrahmer:SqFreeEVAL}, it was shown that the size of the subdivision tree for $\SqFreeEVAL$ is $O(d(L+\ln d))$ for the benchmark problem of isolating all of the real roots of a degree $d$ polynomial with integer coefficients of bit-size at most $L$.  As in the case for Sturm sequences and Descartes' rule of signs, this bound is optimal when $L\geq \ln d$.  In \cite{SagraloffYap:8Point,SagraloffYap:CEVAL}, the authors show that the bit-complexity for a variant of $\SqFreeEVAL$ is $\widetilde{O}(d^4L^2)$.  For other complexity results, see \cite{BurrKrahmerYap:ContinuousAmortization,SagraloffYap:8Point,SagraloffYap:CEVAL}.  The $\SqFreeCEVAL$ algorithm is a variant of the $\SqFreeEVAL$ algorithm, but it can be applied to find the complex roots of a polynomial.  The algorithm was presented in \cite{SagraloffYap:8Point,SagraloffYap:CEVAL}.  There, the authors show that the size of the subdivision tree is $\widetilde{O}(d^2L)$ and the bit-complexity of the algorithm is also $\widetilde{O}(d^4L^2)$.

\subsection{Other methods}
Another common symbolic method for isolating the roots of a polynomial is the continued fraction method and was first presented in \cite{vincentcontinuedfractions}; for more history, see \cite{newVincent}.  The size of the subdivision tree for this algorithm for the benchmark problem of isolating all of the real roots of degree $d$ polynomial with integer coefficients of bit-size at most $L$ is $\widetilde{O}(dL)$ when an ideal root bound is used and $\widetilde{O}(d^2L)$ when a more realistic bound is used, see \cite{Sharma2008292} and the references within.  Also, in \cite{Mehlhorn2010677}, it is shown that the bit-complexity is $\widetilde{O}(d^4L^2)$.  We do not study this algorithm in this paper; a brief discussion of the difficulties can be fond in Section \ref{section:continuousamortization}.

The complexity of all of these methods is worse than the algorithm with bit-complexity $\widetilde{O}(d^3(L+\ln d))$ presented in \cite{schonhage:fundamental}.  This algorithm, however, is not based on subdivision, and it approximates all roots simultaneously.  For more details and other algorithms see the surveys \cite{Pan:Optimal,Pan:History,McNamee:Numerical,McNameePan:Efficient} as well as the discussion in \cite{du-sharma-yap:sturm:07}.

\section{Notation, Continuous Amortization, and the Mahler-Davenport Root Bounds}\label{sec:continuousamortization}

In this paper, our primary focus is on the one- and two-dimensional spaces $\mathbb{R}$ and $\mathbb{C}\simeq\mathbb{R}^2$.  Therefore, our notation is specialized to these cases, and we leave it to the reader to reinterpret some of the details of this paper in higher dimensions.  In one dimension, our basic object are intervals of the form $I=[a,b]$.  The {\em width} of this interval, denoted $w(I)$, is $b-a$, and the {\em midpoint} of this interval, denoted $m(I)$, is $\frac{1}{2}(a+b)$.  To {\em bisect} an interval means that the interval is split at its midpoint into two closed intervals $[a,m(I)]$ and $[m(I),b]$.  In this case, the width of each of these subintervals is half of $w(I)$.

In two dimensions, our basic object are axis-aligned squares of the form $S=I_1\times I_2$ where $w(I_1)=w(I_2)$.  The {\em diameter} of this square, denoted $\diam(S)$, is $\sqrt{2}\cdot w(I_1)$, and is equal to the diameter of the smallest disk which covers $S$.  The {\em midpoint} of a square, denoted by $m(S)$, is the point $(m(I_1),m(I_2))$.  To {\em bisect} a square means that the square is split into four squares by bisecting each of the defining intervals.  This operation is the geometric version of subdivision used in a quad-tree.

For $n$-dimensional spaces, our basic object is a hypercube whose diameter is $\sqrt{n}\cdot w(I)$ where $I$ is one of the edges of the hypercube.  In addition, the coordinates of its midpoint are given by the midpoints of its defining intervals.  Finally, bisection of this hypercube divides it into $2^n$ smaller hypercubes analogous to the subdivision in a $K$-$d$ tree; the diameter of each of these smaller hypercubes is half the diameter of the original hypercube.

Throughout this paper, we consider univariate polynomials $p(x)$ of degree $d$.  In order to achieve complexity bounds, these polynomials are square free and have integral coefficients, i.e., $p(x)\in\mathbb{Z}[x]$.  The {\em height} of such a polynomial is the maximum of the absolute values of the coefficients, and is denoted by $\|p\|$.  For later complexity bounds, we assume that $\|p\|<2^L$, in other words, that the coefficients of $p$ can be written with at most $L$-bits; therefore, we call $L$ the {\em logarithmic height} of $p$.

The problem that we will consider in this paper is to solve the benchmark problem of finding {\em all} of the real or complex roots of a polynomial.  From the constraint that $\|p\|<2^L$, it follows that the magnitude of all of the roots is bounded by $2^L$ \cite{Yap:Algebra}.  Therefore, throughout this paper, we assume that our initial {\em benchmark regions} are $[-2^L,2^L]$ in one dimension and $[-2^L,2^L]\times[-2^Li,2^Li]$ in the complex plane.

\subsection{Continuous Amortization}\label{section:continuousamortization}

The bisection algorithms considered in this paper take, as input, a square-free polynomial\footnote{Even though we only consider square-free polynomials in this paper, bisection algorithm can be extended to more general analytic settings, e.g., see \cite{Chee:Analytic}} $p$ with integer coefficients and initial region $D$ and return a partition of $D$ with each subregion identified as having exactly 0 roots or exactly 1 root.  On a high level, these bisection algorithms are based on a predicate $B$ on intervals in $\mathbb{R}$ or on squares in $\mathbb{C}$.  Consider a subregion $S$ of $D$, if the predicate applied to $B(S)$ is true, then $S$ is terminal, and, if the predicate applied to $S$ is false, then $S$ is bisected and the predicate is applied to its children.  

We will now describe continuous amortization in detail for the one-dimensional setup, as expressed in \cite{BurrKrahmerYap:ContinuousAmortization,BurrKrahmer:SqFreeEVAL}.  Let $I$ be an initial interval and $B$ a Boolean function on subintervals of $I$.  The algorithm constructs a partition $P$ of $I$ such that for each $J$ in $P$, $B(J)$ is true.  Initially, $P=\{I\}$.

\begin{algorithm}{General Bisection Algorithm (cf. \cite{BurrKrahmer:SqFreeEVAL,Mantzaflaris20112312})}\label{algorithm:general}
\noindent Repeatedly bisect each $J\in P$ until the following condition holds:\\[.15cm]
\indent $B(J)$ is true.
\end{algorithm}

We then call a nonnegative function $F:I\rightarrow\mathbb{R}$ a {\em stopping function} if it has the following property: for any pair $(x,J)$, with $J$ a subinterval of $I$ containing the point $x$ in $I$; if $w(J)<F(x)$, then $B(J)$ is true.  Therefore, $F(x)$ is a lower bound on the size of a non-terminal subinterval of $I$ containing $x$.  With this setup, we can apply the Continuous Amortization Theorem to compute the size of the partition.

\begin{lemma}[Continuous Amortization \cite{BurrKrahmerYap:ContinuousAmortization,BurrKrahmer:SqFreeEVAL}]\label{lemma:continuousamortization}
Let $F$ be a stopping function for Algorithm \ref{algorithm:general}, and let $Q(I)$ be the final partition produced by the algorithm when applied to the interval $I$.  Then,
$$
\#Q(I)\leq\max\left\{1,\int_I\frac{2dx}{F(x)}\right\}.
$$
If the algorithm does not terminate, then the integral is infinite.
\end{lemma}

In this setting, $\#Q(I)$ is the number of leaves of the subdivision tree; therefore, the size of the subdivision tree is $2\#Q(I)-1$.  For the cases being considered in this paper, i.e., for the isolation of the roots of $p$, the predicate $B$ consists of two tests: {\em inclusion} and {\em exclusion predicates}.  When an exclusion predicate is true on an interval $J$, then there are no roots of $p$ in $J$, and when an inclusion predicate is true on an interval $J$, then there is exactly one root of $p$ in $J$.  If both the inclusion and exclusion predicates fail, then $B(J)$ is false, and $J$ is bisected.  For real root isolation, these predicates are often applied to open intervals; to avoid missing roots, one must evaluate the polynomial $p$ at midpoints to check for roots on the boundaries of these intervals.  Therefore, the terminal intervals isolate the roots of $p$.  In Sections \ref{sec:complex} and \ref{sec:bit} we will show how to extend the continuous amortization theorem to $\mathbb{R}^n$ and to compute the bit-complexity of an algorithm.

At this point we note that as presented above, continuous amortization is not formulated for the continued fraction algorithm for root isolation.  In particular, the continued fraction algorithm is not based on bisection, e.g., the positive lower bound (\texttt{PLB}) on the smallest positive real roots can split off intervals of varying sizes.  It is possible to adapt continuous amortization to the continued fraction algorithm, but we do not study that possibility here.

\subsection{Mahler-Davenport Bound}

In each of the computations below, the final step in bounding the integral developed from applying continuous amortization requires bounding the sum of logarithms of distances between roots.  Such a bound is given by the Mahler-Davenport bound

\begin{lemma}[Mahler-Davenport root separation bound \cite{eigenwillig-sharma-yap:descartes:06}]\label{lemma:mahlerdavenport}
Let $p$ be a square-free complex polynomial of degree $d$ with roots $V=\{\alpha_1,\cdots,\alpha_d\}$.  Let $G=(V,E)$ be a directed graph on the roots of $p$ where $E=\{(\alpha_{i_1},\alpha_{j_1})\}$ such that 
\begin{inparaenum}[(1)]
\item the directed edges point in the direction of decreasing magnitude,
\item the graph is acyclic, and 
\item the in-degree of any node is at most 1.
\end{inparaenum}
Then, 
$$
\prod_{(\alpha_{i_1},\alpha_{j_1})\in E}|\alpha_{i_1}-\alpha_{j_1}|\geq\sqrt{|\disc(p)|}\cdot M(p)^{-(d-1)}\cdot\left(\frac{d}{\sqrt{3}}\right)^{-|E|}\cdot d^{-d/2}.
$$
Here $|E|$ is the number of edges in the graph, $\disc(p)$ is the discriminant of $p$, $M(p)$ is the Mahler measure of $p$, i.e., $M(p)=\lead(p)\prod_{i=1}^d\max\{1,|\alpha_i|\}$, and $\lead(p)$ is the leading coefficient of $p$.
\end{lemma}

In the cases below, we use a bound on the negative of the logarithm of the distances between roots.  In this case, the Mahler-Davenport bound reduces to
$$
\sum_{(\alpha_{i_1},\alpha_{j_1})\in E}-\ln(|\alpha_{i_1}-\alpha_{j_1}|)\leq-\frac{1}{2}\ln(|\disc(p)|)+(d-1)\ln(M(p))+|E|(\ln d-\ln\sqrt{3})+\frac{d}{2}\ln d.
$$
In the situation discussed above, i.e., where $p$ is restricted to have integer coefficients and logarithmic height $L$, the Mahler-Davenport bound can be simplified.  In particular, since $p$ is square free and $p$ has integer coefficients, the discriminant of $p$ is a nonzero integer.  Therefore, $-\frac{1}{2}\ln(|\disc(p)|)$ is a nonpositive number.  In addition, the Mahler measure $M(p)$ is bounded by $\|p\|_2$, the $2$-norm of the coefficients of $p$, see \cite{Yap:Algebra}.  In turn, the $2$-norm is bounded by $\sqrt{d+1}\cdot 2^L$.  It follows that for $d\geq 1$, $(d-1)\ln(M(p))\leq dL+\frac{d}{2}\ln d$.  Combining these facts, we have that 
\begin{equation}\label{equation:MahlerDavenport}
\sum_{(\alpha_{i_1},\alpha_{j_1})\in E}-\ln(|\alpha_{i_1}-\alpha_{j_1}|)\leq dL+(d+|E|)\ln d.
\end{equation}
Finally, in the Mahler-Davenport bound, since the in-degree of each node is at most $1$, there are at most $d$ edges, so Inequality \ref{equation:MahlerDavenport} is bounded above by $dL+2d\ln d$.  To simplify our use of the Mahler-Davenport bound below, we use the following corollary:

\begin{corollary}\label{corollary:MahlerDavenport}
Let $p$ be a square-free polynomial with integer coefficients of degree $d$ and roots $V=\{\alpha_1,\cdots,\alpha_d\}$.  Let $G=(V,E)$ on the roots of $p$ such that the valence of any root is bounded by $k$, then 
$$
\sum_{(\alpha_{i_1},\alpha_{j_1})\in E}-\ln(|\alpha_{i_1}-\alpha_{j_1}|)\leq kdL+(kd+|E|)\ln d=O(kdL+kd\ln d).
$$
\end{corollary}
\noindent{\em Proof Sketch.} Begin by orienting each of the edges of $G$ according to the conditions in Lemma \ref{lemma:mahlerdavenport}.  Since the valence of any root is bounded by $k$, the in-degree of any root is also bounded by $k$.  For each root, label its incident edges starting with the value 1.  Let $G_i$ be the subgraph of $G$ consisting of all edges of $G$ with label $i$.  Each $G_i$ satisfies the conditions of the Mahler-Davenport bound, so the result above is achieved by applying the Mahler-Davenport bound at most $k$ times.

\section{Real Root Isolation}\label{sec:real}
In this section, we apply the continuous amortization technique to the most common bisection-based symbolic techniques for real root isolation.  In particular, we study Sturm sequences \cite{CollinsLoos:Sturm} and Descartes' rule of signs \cite{CollinsAkritas:Descartes}.  Our results match the best bounds in the literature.

\subsection{Sturm Sequences}\label{sec:real:sturm}
One of the most powerful and well-known methods to isolate the real roots of a univariate polynomial is to use Sturm sequences.  The Sturm sequence for a square-free\footnote{Note that Sturm sequences can be applied to polynomials which are not square-free, see \cite{Yap:Algebra}, but we do not consider that option here.} and univariate polynomial $p$ is the following sequence of polynomials: $p_0:=p$, $p_1:=p'$, and $p_i:=-\text{Remainder}(p_{i-2},p_{i-1})$ (let $N$ be the step where this sequence terminates).  At a particular point $a$, the variation in this sequence is the number of sign changes of the sequence $(p_0(a),\cdots,p_N(a))$ (zeros do not count as sign changes).

\begin{lemma}[Sturm's Theorem \cite{SturmOriginal}]
Let $p$ be a square-free univariate polynomial and $(a,b]$ a real interval.  The number of distinct real roots in the interval is the difference between the number of sign changes at $a$ and $b$.
\end{lemma}

Sturm sequences can be used to create inclusion and exclusion predicates in the following way.  An interval $J$ is excluded if both of its endpoints have the same variation; in this case, $J$ contains no roots.  Similarly, an interval $J$ is included if the difference in variation between its endpoints is exactly 1; in this case, $J$ contains exactly one root.  Therefore, our predicate $B_{\Sturm}$ for Sturm sequences is based on the number of sign changes of the Sturm sequence at the endpoints of $J$, or, equivalently, the number of real roots in $J$.  
$$
B_{\Sturm}(J)=\begin{cases}\true&\text{Variation $=$ 0 or 1}\\\false&\text{Variation $\geq$ 2}\end{cases}.
$$

We next derive a stopping function for the Sturm predicate.  Let the real roots of $p$ be $\alpha_1<\cdots<\alpha_k$, and let $x$ be any point in the benchmark interval $I=[-2^L,2^L]$.  Define $\dist_2(x,\{\alpha_i\})$ to be the distance between $x$ and the second closest real root of $p$.  We derive a lower bound on the width of an interval $J$ that contains $x$ as well as two roots (and, therefore, is not terminal).  Suppose that $J$ contains both $x$ as well as two roots $\alpha_{i_1}$ and $\alpha_{i_2}$.  Assume w.l.o.g. that $\alpha_{i_1}$ is not further from $x$ than $\alpha_{i_2}$, i.e., $\dist(x,\alpha_{i_1})\leq\dist(x,\alpha_{i_2})$.  The value of $\dist(x,\alpha_{i_1})$ the distance between $x$ and the closest real root of $p$, while $\dist(x,\alpha_{i_2})=\dist_2(x,\{\alpha_i\})$.  Since $J$ contains both $x$ and $\alpha_{i_2}$, it is necessary that $w(J)>\dist_2(x,\{\alpha_i\})$.  Therefore, our stopping function will restrict $J$ by $\dist_2(x,\{\alpha_i\})$.  In particular, let $F_{\Sturm}$ be the stopping function for the Sturm sequence algorithm, then
$$
F_{\Sturm}=\dist_2(x,\{\alpha_i\}).
$$

\begin{remark}
Note that, in general, picking a stopping function is a very delicate process.  The stopping function must be simple enough to be integrable, but precise enough to lead to an interesting bound on the complexity of the algorithm.  The stopping function defined above is not the largest possible stopping function for $B_{\Sturm}$ (larger stopping functions are better because the corresponding bound on widths of intervals is more relaxed).  One could simply state that the stopping function at $x$ is the width of the smallest interval containing $x$ and two roots.  This stopping function may be larger than the function above when $x$ is between the two closest real roots and the third closest real root is also near, e.g., when  $\alpha_{i_2}$ and $\alpha_{i_2-1}$ are closer to each other than $x$ is to $\alpha_{i_1}$, see Figure \ref{sturm:smallinterval}.  While the ideal stopping function is a better bound on the depth of the tree, it is more complicated to integrate.
\begin{figure}[htb]
\center
\begin{tikzpicture}
\draw (0,0) node[below] {$x$};
\draw (-2.25,0) -- (1.25,0);
\filldraw (1,0) circle (2pt) node[below]{$\alpha_{i_1}$}
	 (-1.25,0) circle (2pt) node[below]{$\alpha_{i_2}$}
	 (-2,0) circle (2pt)node[below] {$\alpha_{i_2-1}$};
\draw (-.1,-.1) -- (.1,.1) 
	(-.1,.1) -- (.1,-.1);
\end{tikzpicture}
\caption{The smallest interval containing $x$ and two roots does not contain the two closest roots to $x$.  Instead, it contains the second and third closest roots to $x$.  Therefore, a function which returns the smallest interval containing $x$ and two other roots can be complicated.}
\label{sturm:smallinterval}
\end{figure}
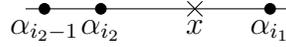
\end{remark}

Continuing with the computation, we must integrate the reciprocal of the stopping function $F_{\Sturm}$; we proceed by explicitly describing the stopping function.  For each root $\alpha_i$, define $I_{\alpha_i}$ to be the set of points of $I$ which are closest to $\alpha_i$, i.e., the one-dimensional Voronoi cell of $\alpha_i$.  In this case,
$$
I_{\alpha_i}=\begin{cases}\left[-2^L,(\alpha_1+\alpha_2)/2\right]&i=1\\
\left[(\alpha_{i-1}+\alpha_i)/2,(\alpha_i+\alpha_{i+1})/2\right]&1<i<k\\
\left[(\alpha_{k-1}+\alpha_k/2,2^L)\right]&i=k
\end{cases}.
$$
For any $x\in I_{\alpha_i}$, the second closest real root is either $\alpha_{i-1}$ or $\alpha_{i+1}$ (provided such roots exist).  Define $I_{\alpha_i,L}$ to be the subinterval of $I_{\alpha_i}$ consisting of points whose second closest real root is to the left of $\alpha_i$, i.e., the second closest real root is $\alpha_{i-1}$.  Similarly, define $I_{\alpha_i,R}$ to be the points whose closest real root is $\alpha_i$ and whose second closest real root is to the right of $\alpha_i$.  Note that $I_{\alpha_1}=I_{\alpha_1,R}$ since there are no real roots to the left of the smallest real root, and $I_{\alpha_k}=I_{\alpha_k,L}$ since there are no roots to the right of $\alpha_k$.  Note that for $1<i<k$, $I_{\alpha_i}$ is split half-way between the neighboring roots, i.e., $(\alpha_{i-1}+\alpha_{i+1})/2$.  Therefore,
\begin{align*}
I_{\alpha_i,L}=\begin{cases}\emptyset&i=1\\
\left[(\alpha_{i-1}+\alpha_i)/2,(\alpha_{i-1}+\alpha_{i+1})/2\right]&1<i<k\\
\left[(\alpha_{k-1}+\alpha_k)/2,2^L\right]&i=k
\end{cases}\\
I_{\alpha_i,R}=\begin{cases}\left[-2^L,(\alpha_1+\alpha_2)/2)\right]&i=1\\
\left[(\alpha_{i-1}+\alpha_{i+1})/2,(\alpha_i+\alpha_{i+1})/2\right]&1<i<k\\
\emptyset&i=k
\end{cases}.
\end{align*}
For a point in $I_{\alpha_i,L}$ the second nearest real root is $\alpha_{i-1}$, and the distance between $x$ and this root is $x-\alpha_{i-1}$.  For a point in $I_{\alpha_i,R}$, the second nearest real root is $\alpha_{i+1}$, and the distance between $x$ and this root is $\alpha_{i+1}-x$.  Therefore, we can rewrite the stopping function $F_{\Sturm}$ as
$$
F_{\Sturm}(x)=
\begin{cases}
x-\alpha_{i-1}&x\in I_{\alpha_i,L}\\
\alpha_{i+1}-x&x\in I_{\alpha_i,R}
\end{cases}.
$$

Let $Q_{\Sturm}$ be the partition of $I$ at the end of the Sturm bisection algorithm.  Using continuous amortization and the definition for the stopping function $F_{\Sturm}$, we can bound the size\footnote{Throughout the remainder of this paper, we will often be bounding pairs of integrals as below.  To make referencing easier, the equations are labeled with an ordered pair so that the integrals on the right-hand-side of the equality can be referenced easily and individually.} of this partition as follows:
\begin{subequations}\label{equation:SturmReal}
\begin{equation}
\#Q_{\Sturm}\leq \int_I{\frac{2dx}{F_{\Sturm}(x)}}=\sum_{i=1}^{k-1}\int_{I_{\alpha_i},R}{\frac{2dx}{\alpha_{i+1}-x}}+\sum_{i=2}^k\int_{I_{\alpha_i},L}{\frac{2dx}{x-\alpha_{i-1}}}\eqlabel
\end{equation}
\end{subequations}   
Next, we evaluate these integrals.  For $2\leq i\leq k-1$, Integral \ref{equation:SturmReal}a can be evaluated as
$$
\int_{I_{\alpha_i},R}{\frac{2dx}{\alpha_{i+1}-x}}=\int_{\frac{1}{2}(\alpha_{i-1}+\alpha_{i+1})}^{\frac{1}{2}(\alpha_i+\alpha_{i+1})}{\frac{2dx}{\alpha_{i+1}-x}}=-2\ln(\alpha_{i+1}-\alpha_i)+2\ln(\alpha_{i+1}-\alpha_{i-1}).
$$
Similarly, Integral \ref{equation:SturmReal}b evaluates to 
$$
\int_{I_{\alpha_i},L}{\frac{2dx}{x-\alpha_{i-1}}}=\int_{\frac{1}{2}(\alpha_{i-1}+\alpha_i)}^{\frac{1}{2}(\alpha_{i-1}+\alpha_{i+1})}{\frac{2dx}{x-\alpha_{i-1}}}=2\ln(\alpha_{i+1}-\alpha_{i-1})-2\ln(\alpha_i-\alpha_{i-1}).
$$
For $i=1$, 
$$
\int_{I_{\alpha_1},R}{\frac{2dx}{\alpha_2-x}}=\int_{-2^L}^{\frac{1}{2}(\alpha_1+\alpha_2)}{\frac{2dx}{\alpha_2-x}}=-2\ln(\alpha_2-\alpha_1)+2\ln(\alpha_2+2^L),
$$
and, for $i=k$,
$$
\int_{I_{\alpha_k},L}{\frac{2dx}{x-\alpha_{k-1}}}=\int_{\frac{1}{2}(\alpha_{k-1}+\alpha_k)}^{2^L}{\frac{2dx}{x-\alpha_{k-1}}}=2\ln(2^L-\alpha_{k-1})-2\ln(\alpha_k-\alpha_{k-1}).
$$
Combining these calculations gives us an upper bound on the size of the partition.  In other words,
$$
\#Q_{\Sturm}\leq 2\ln(\alpha_2+2^L)+2\ln(2^L-\alpha_{k-1})+\sum_{i=1}^{k-2}4\ln(\alpha_{i+2}-\alpha_i)-\sum_{i=1}^{k-1}4\ln(\alpha_{i+1}-\alpha_i).
$$
Since the maximum absolute value of the roots is $2^L$, each of the logarithms with a positive coefficient is bounded by $\ln(2^{L+1})=(L+1)\ln2$.  There are $4(k-1)$ terms of this form; therefore, the terms with positive coefficient are bounded above by $4\ln(2)(k-1)(L+1)$.  Since the real roots form a subset of all roots, and $k$ is the number of real roots, $k\leq d$ so the terms with positive coefficient are $O(dL)$.

The terms with negative coefficient can be bounded using the Mahler-Davenport bound of Lemma \ref{lemma:mahlerdavenport}.  We construct a graph on the real roots of $p$ by connecting roots in increasing order.  This graph has maximum degree 2 and the negative logarithm of the lengths of these edges are exactly what we must compute.  From Corollary \ref{corollary:MahlerDavenport}, we know that these terms are bounded by $O(d(L+\ln d))$.  Combining these bounds shows that the size of the subdivision tree for Sturm sequences is $O(d(L+\ln d))$, matching the bounds in \cite{davenport:85,du-sharma-yap:sturm:07}.
 
\subsection{Descartes' Rule of Signs}\label{sec:Descartes}
Descartes' rule of signs provides a less powerful, but still well-known method to isolate the real roots of a univariate polynomial.  Similar to Sturm sequences, Descartes' rule of signs is based on the number of sign changes in a sequence, but, for Descartes' rule of signs, the sequence of values consists of the coefficients of a polynomial.  In particular
\begin{lemma}[Descartes' rule of signs, see \cite{Krandick:2005:NBD:1113439.1113453}]
Let $p$ be a univariate polynomial.  The number of positive real roots (counted with multiplicities) is bounded above by the number of sign changes in the coefficients of $p$.  The difference between the number of sign changes and the actual number of roots is an even integer.
\end{lemma}
Even though Descartes' rule of signs is formulated for positive roots, there are M\"obius transformations, such as $\frac{ax+b}{x+1}$, which transform the positive real axis into the interval $(a,b)$.  Therefore, it makes sense to talk about applying Descartes' rule of signs to an interval.  In this section, when we discuss the {\em variation of the coefficients of a polynomial} $p$ on an interval $(a,b)$, we mean the variation in the coefficients of $(x+1)^d\cdot p\left(\frac{ax+b}{x+1}\right)$.

Descartes' rule of signs can be used to create inclusion and exclusion predicates in the following way.  An interval $J$ is excluded if there is no variation in the coefficients of $p$ on $J$; in this case, $J$ contains no roots.  Similarly, an interval $J$ is included if the variation in the coefficients of $p$ on $J$ is one; in this case, $J$ contains exactly one root.  Therefore, our predicate $B_{\Descartes}$ for Descartes' rule of signs is based on the variation of the coefficients of $p$ on $J$.  In particular
$$
B_{\Descartes}(J)=\begin{cases}\true&\text{Coefficient variation on $J$ = 0 or 1}\\
\false&\text{Coefficient variation on $J$ $\geq$ 2}\end{cases}.
$$

We next derive a stopping function for the Descartes predicate.  The stopping function is based on two theorems: the One- and Two-circle Theorems \cite{Ostrowski1950,Krandick:2005:NBD:1113439.1113453,eigenwillig-sharma-yap:descartes:06}.  These two theorems give conditions on when an interval is terminal under Descartes' rule of signs.  In particular
\begin{lemma}[One-circle Theorem \cite{Ostrowski1950,Krandick:2005:NBD:1113439.1113453,eigenwillig-sharma-yap:descartes:06}]\label{lemma:onecircle}
Let $p$ be a univariate polynomial and $J$ a real interval.  If there are no roots in the open circle with diameter $J$, then the variation in the coefficients of $p$ on $J$ is zero.  Therefore, $B_{\Descartes}(J)=\true$.
\end{lemma}
After the above M\"obius transformation, the One-circle Theorem corresponds to the case where there are no roots with positive real part.
\begin{lemma}[Two-circle Theorem \cite{Ostrowski1950,Krandick:2005:NBD:1113439.1113453,eigenwillig-sharma-yap:descartes:06}]\label{lemma:twocircle}
Let $p$ be a univariate polynomial and $J$ a real interval.  If there is exactly one root in the two circles which circumscribe equilateral triangles with one side $J$, then the variation in the coefficients of $p$ on $J$ is one.  Therefore, $B_{\Descartes}(J)=\true$.
\end{lemma}
After the M\"obius transformation from above, the Two-circle Theorem corresponds to the case where there is exactly one root in the portion of the plane subtended by an arc of $\frac{4}{3}\pi$.  Note also that the diameter of the region covered by the two circles is $\sqrt{3}\cdot w(J)$.

We can use these two lemmas to develop a stopping function for Descartes' rule of signs.  Let the roots of $p$ be $\{\alpha_1,\cdots,\alpha_d\}$; w.l.o.g., reorder these roots so that $\alpha_1,\cdots,\alpha_k$ are the real roots of $p$ while $\alpha_{k+1},\cdots,\alpha_d$ are the complex roots of $p$.  Let $x$ be any point in the benchmark interval $I=[-2^L,2^L]$; define $\dist(x,\{\alpha_i\})$ to be the distance from $x$ to the nearest root of $p$, and $\dist_2(x,\{\alpha_i\})$ to be the distance from $x$ to the second closest root of $p$.  (Note that the definitions in this section differ from the definitions in the Sturm sequence argument because they are not restricted to real roots.)  Let $J$ be an interval containing $x$; if $w(J)<\dist(x,\{\alpha_i\})$, then the open disk with diameter $J$ is contained within the open disk centered at $x$ with radius $\dist(x,\{\alpha_i\})$.  By construction, the open disk centered at $x$ with radius $\dist(x,\{\alpha_i\})$ contains no roots of $p$; therefore, the open disk with diameter $J$ also does not contain any roots of $p$.  Therefore, by the One-circle Theorem, Lemma \ref{lemma:onecircle}, $J$ is terminal, and $F_{\Descartes,1}(x)=\dist(x,\{\alpha_i\})$ is a stopping function for Descartes' rule of signs.

Suppose now that $J$ is an interval containing $x$ such that $w(J)<\frac{1}{\sqrt{3}}\dist_2(x,\{\alpha_i\})$, then the open disk circumscribing the two disks in the Two-circle Theorem, Lemma \ref{lemma:twocircle}, is a disk of diameter less than $\dist_2(x,\{\alpha_i\})$.  Thus this disk contains at most one root; moreover, both the open disk of the One-circle Theorem and the two open disks of the Two-circle Theorem are contained in the disk centered at $x$ of radius $\dist_2(x,\{\alpha_i\})$.  It follows that either the open disk of the One-circle Theorem contains no roots or the two open disks of the Two-circle Theorem contain at most one root.  Thus by either Lemmas \ref{lemma:onecircle} or \ref{lemma:twocircle}, $J$ is terminal.  Therefore, $F_{\Descartes,2}(x)=\frac{1}{\sqrt{3}}\dist_2(x,\{\alpha_i\})$ is a stopping function for Descartes' rule of signs.

We combine the two stopping functions $F_{\Descartes,1}$ and $F_{\Descartes,2}$ into a single stopping function for Descartes' rule of signs; since both functions are stopping functions, we are free to choose the most appropriate stopping function at each point.  Near real roots, $F_{\Descartes,1}$ tends to infinity, so we use $F_{\Descartes,2}$ near real roots of $p$ and $F_{\Descartes,1}$ at all other points.  In particular, for each root $\alpha_i$ of $p$, define $I_{\alpha_i}$ to be the set of points of $I$ which are closest to $\alpha_i$.  For each real root $\alpha_i$ with $1\leq i\leq k$, define $d_i$ to be the distance from $\alpha_i$ to the nearest distinct root of $p$, and let $J_{\alpha_i}=\left[\alpha_i-\frac{1}{1+\sqrt{3}}d_i,\alpha_i+\frac{1}{1+\sqrt{3}}d_i\right]$.  (The reason for the coefficient of $d_i$ will be indicated later in the computation.)  Let $J$ be the union of the $J_{\alpha_i}$'s; then $J$ is the region where the stopping function $F_{\Descartes,2}$ is used.  In other words, the stopping function for Descartes' rule of signs is
$$
F_{\Descartes}(x)=\begin{cases}F_{\Descartes,1}(x)=\dist(x,\{\alpha_i\})&x\in I\setminus J\\
F_{\Descartes,2}(x)=\dist_2(x,\{\alpha_i\})&x\in J\end{cases}.
$$

Now, we are ready to apply continuous amortization to compute the complexity of the subdivision tree for Descartes' rule of signs.  Let $Q_{\Descartes}$ be the partition of $I$ at the end of the Descartes' rule of signs bisection algorithm.  Using continuous amortization on the stopping function $F_{\Descartes}$ results in the following inequalities:
\begin{subequations}\label{equation:descartes}
\begin{equation}
\#Q_{\Descartes}\leq\int_I{\frac{2dx}{F_{\Descartes}}}\leq \int_{I\setminus J}{\frac{2dx}{F_{\Descartes,1}(x)}}+\int_{J}{\frac{2dx}{F_{\Descartes,2}(x)}}.\eqlabel
\end{equation}
\end{subequations}

We next evaluate this integral for each root beginning with the real roots.  Let $\alpha_j$ be a real root (by assumption, $1\leq j\leq k$) and let $I_{\alpha_j}=[a_j,b_j]$.  The portion of Inequality \ref{equation:descartes} over this interval is 
\begin{subequations}\label{equation:descartes:real}
\begin{equation}
\int_{I_{\alpha_j}\setminus J_{\alpha_j}}{\frac{2dx}{F_{\Descartes,1}(x)}}+\int_{J_{\alpha_j}}{\frac{2dx}{F_{\Descartes,2}(x)}}.\eqlabel
\end{equation}
\end{subequations}
Since these integrals are taken over $I_{\alpha_j}$, the closest root to any point $x$ is $\alpha_j$, which is also real.  Therefore, Integral \ref{equation:descartes:real}a can be evaluated as
\begin{align*}
\int_{I_{\alpha_j}\setminus J_{\alpha_j}}{\frac{2dx}{F_{\Descartes,1}(x)}}&=\int_{a_j}^{\alpha_j-\frac{1}{1+\sqrt{3}}d_j}{\frac{2dx}{\alpha_j-x}}+\int_{\alpha_j+\frac{1}{1+\sqrt{3}}d_j}^{b_j}{\frac{2dx}{x-\alpha_j}}\\
&=2\ln(\alpha_j-a_j)+2\ln(b_j-\alpha_j)-4\ln\left(\frac{1}{1+\sqrt{3}}d_j\right).
\end{align*}
Since the maximum magnitude of all roots and points in the interval $I$ is $2^L$, each of the logarithms with a positive coefficient is bounded by $\ln(2^{L+1})=(L+1)\ln 2$.  There are $2k$ terms of this form; therefore, the terms with positive coefficient are bounded  above by $4\ln(2)k(L+1)$.  Since the real roots form a subset of all roots, and $k$ is the number of real roots, $k\leq d$, the terms with positive coefficient are $O(dL)$.  In addition, separating out the $(1+\sqrt{3})$ part of the third logarithm, we have $4\ln(1+\sqrt{3})$, which occurs $k$ times.  Since $k\leq d$, this sum is $O(d)$.  Finally, we are left with a term of the form $-4\ln(d_j)$, which is summed over all $k$ real roots.  Since this is a sum of the negative of logarithms of the distances between roots, the Mahler-Davenport bound of Lemma \ref{lemma:mahlerdavenport} applies.  In particular, we can construct a graph by connecting each real root with its nearest neighbor; in this case, the maximum valence of a vertex is three (each of the neighbors contribute one as well as the root itself).  Then, by Corollary \ref{corollary:MahlerDavenport}, we know that this sum is $O(d(L+\ln d))$.
 
In Integral \ref{equation:descartes:real}b, let $\alpha_k$ be the second closest root to $x$.  Then 
\begin{multline}\label{eq:Descartes}
F_{\Descartes,2}(x)=\frac{1}{\sqrt{3}}\dist_2(x,\{\alpha_i\})=\frac{1}{\sqrt{3}}|\alpha_k-x|\\
\geq\frac{1}{\sqrt{3}}\left(|\alpha_j-\alpha_k|-|\alpha_j-x|\right)\geq \frac{1}{\sqrt{3}}\left(d_j-|\alpha_j-x|\right).
\end{multline}
By working backwards from this right-hand side of the above inequality, one can reach our choice of $\frac{1}{1+\sqrt{3}}$ above.  Using this bound,
\begin{multline*}
\int_{\alpha_j-\frac{1}{1+\sqrt{3}}d_j}^{\alpha_j+\frac{1}{1+\sqrt{3}}d_j}{\frac{2dx}{F_{\Descartes,2}(x)}}\leq\int_{\alpha_j-\frac{1}{1+\sqrt{3}}d_j}^{\alpha_j+\frac{1}{1+\sqrt{3}}d_j}{\frac{2\sqrt{3}\cdot dx}{d_j-|\alpha_j-x|}}\\
=4\sqrt{3}\ln(d_j)-4\sqrt{3}\ln\left(\frac{\sqrt{3}}{1+\sqrt{3}}d_j\right)=-4\sqrt{3}\ln\left(\frac{\sqrt{3}}{1+\sqrt{3}}\right).
\end{multline*}
This term occurs for each of the $k$ roots and since $k\leq d$, this sum is $O(d)$.  Combining all of these bounds, we have that the real part is $O(d(L+\ln d))$.

The last step of the computation is to compute the integral for complex roots.  Let $\alpha_j$ be a complex root with positive imaginary part (by assumption $k<j\leq d$).  We need only consider complex roots with positive imaginary part because $I_{\alpha_j}=I_{\overline{\alpha_j}}$ and it is only necessary to integrate over each interval once.  As above, let $I_{\alpha_j}=[a_j,b_j]$.  Then we bound Inequality \ref{equation:descartes} (using the fact that $\arcsinh$ is an odd function):
\begin{multline*}
\int_{I_{\alpha_j}}{\frac{2dx}{F_{\Descartes,1}(x)}}=\int_{a_j}^{b_j}\frac{2dx}{|\alpha_j-x|}=2\cdot\arcsinh\left(\frac{b_j-\real(\alpha_j)}{\imag(\alpha_j)}\right)+2\cdot\arcsinh\left(\frac{\real(\alpha_j)-a_j}{\imag(\alpha_j)}\right)\\
=2\ln(b_j-\real(\alpha_j)+|b_j-\alpha_j|)+2\ln(\real(\alpha_j)-a_j+|a_j-\alpha_j|)-4\ln(\imag(\alpha_j)).
\end{multline*}
Since the maximum magnitude of all roots and points in the interval $I$ is $2^L$, each of the logarithms with a positive coefficient is bounded by $\ln(2^{L+2})=(L+2)\ln 2$.  There are $d-k$ terms of this form; therefore, the terms with positive coefficient are bounded above by $2\ln(2)(d-k)(L+1)$.  Since $k$ is a nonnegative integer, $d-k\leq d$, and the terms with positive coefficient are $O(dL)$.  The remaining terms with negative coefficient form a sum of the negative of logarithms of the distances between roots since $\ln(\imag(\alpha_j))=\ln((\alpha_j-\overline{\alpha_j})/2)$.  Separating out the division by 2 gives that these terms are bounded by $-4\ln(\alpha_j-\overline{\alpha_j})+4\ln2$.  Since there are $(d-k)/2$ terms of this form, the terms with $\ln2$ sum to $2(d-k)\ln2$, and, since $(d-k)/2\leq d$, this sum is $O(d)$.  This leaves the sum of the negative of the logarithms of the distances between conjugate pairs of roots, so the Mahler-Davenport bound of Lemma \ref{lemma:mahlerdavenport} applies.  In particular, we can construct a graph by connecting each complex root to its conjugate pair.  In this graph, the maximum valence is 1, therefore, by Corollary \ref{corollary:MahlerDavenport}, this sum is $O(d(L+\ln d))$.  Combining all of these bounds shows that the size of the subdivision tree is $O(d(L+\ln d))$ matching the bound of \cite{eigenwillig-sharma-yap:descartes:06}.

\section{Complex Root Isolation}\label{sec:complex}
In this section, we extend the continuous amortization technique to bisection-based techniques for complex root isolation.  In particular, we present a generalization of continuous amortization to arbitrary dimensions, and we use the extension to compute the size of the subdivision tree for Sturm sequences \cite{Pinkert:1976:EMF:355705.355710,Wilf:1978:GBA:322077.322084,CamargoBrunetto200095,du-sharma-yap:sturm:07} and a square-free version of \texttt{CEVAL} \cite{SagraloffYap:8Point}.  For the purposes of this section, we often view $\mathbb{C}$ as $\mathbb{R}^2$.  Our complexity bounds match or surpass the best bounds in the literature.

\subsection{Generalizing Continuous Amortization to Higher Dimensions}\label{sec:CAHigher}
In this section, we consider bisection-based algorithms on $\mathbb{R}^n$ for a fixed dimension $n$.  Our basic object of interest is a hypercube of dimension $n$, and to bisect it means to divide it in half along each dimension, which results in $2^n$ smaller hypercubes whose $n$-dimensional volume is $2^{-n}$ times the $n$-dimensional volume of the original hypercube.  For $n$-dimensional problems, the predicate $B$ is a Boolean function on hypercubes.  In this setup, Algorithm \ref{algorithm:general} applies as written, while using the higher-dimensional version of bisection.  In this case, we call a nonnegative function $F:\mathbb{R}^n\rightarrow\mathbb{R}$ a stopping function if it has the following property: for any pair $(x,J)$ where $x$ is a point in the hypercube $J$, if the $n$-dimensional volume of $J$ is less than $F(x)$, then $B(J)$ is true.  With this setup, we extend continuous amortization to compute the size of the partition formed from Algorithm \ref{algorithm:general} in this higher-dimensional situation.

\begin{proposition}\label{prop:ndimensional}
Let $F$ be a stopping function for Algorithm \ref{algorithm:general} in $n$-dimensional space, and let $Q(R)$ be the final partition produced by the algorithm when applied to the $n$-dimensional hypercube $R$.  Then,
$$
\#Q(R)\leq\max\left\{1,\int_R{\frac{2^ndV}{F(x)}}\right\},
$$
where $dV$ is the $n$-dimensional volume form.  If the algorithm does not terminate, then the integral is infinite.
\end{proposition} 
\begin{proof}
In this proof, we follow the proof as in \cite{BurrKrahmer:SqFreeEVAL}.  If $\#Q(R)=1$, then the bound is immediate.  If $\#Q(R)>1$, then let $S$ be any interval in the partition $Q(R)$.  There is a lower bound on the $n$-dimensional volume of $S$ because Algorithm \ref{algorithm:general} did not terminate on the parent of $S$.  Therefore, 
$$
\forall y\in S, \Vol_n(S)\geq\frac{1}{2^n}F(y),
$$
where $\Vol_n(S)$ is the $n$-dimensional volume of $S$.  On the other hand, since
$$
\int_R{\frac{2^ndV}{F(x)}}=\sum_{S\in Q(R)}\int_S{\frac{2^ndV}{F(x)}},
$$
it suffices to show that the value on each terminal hypercube is at least one.  For a fixed $S$, let $z\in S$ such that $F(z)$ is maximal in $S$, then
$$
\int_S{\frac{2^ndV}{F(x)}}\geq\int_S{\frac{2^ndV}{F(z)}}\geq\frac{2^n}{F(z)}\cdot \Vol_n(S)\geq \frac{2^n}{F(z)}\cdot\frac{F(z)}{2^n}=1.
$$
If Algorithm \ref{algorithm:general} does not terminate, then the integral above is a bound for the size of the partition at any point in time, but since the size of the partition can be arbitrarily large, the integral is unbounded.
\end{proof}

Proposition \ref{prop:ndimensional} is of theoretical and practical interest, and, in this paper, we apply it in the two-dimensional case for $\mathbb{C}\simeq\mathbb{R}^2$.  Here, $n=2$ and $F(x)$ gives a bound on the area of a subregion.  The following corollary explicitly expresses this special case:

\begin{corollary}
Let $F$ be a stopping function for Algorithm \ref{algorithm:general} in $\mathbb{C}$ (or $\mathbb{R}^2$), and let $Q(R)$ be the final partition produced by the algorithm when applied to the square $R$.  Then,
$$
\#Q(R)\leq\max\left\{1,\int_R{\frac{4dA}{F(x)}}\right\},
$$
where $dA$ is the area form.  If the algorithm does not terminate, then the integral is infinite.
\end{corollary}

In the next two sections, we apply this generalization of continuous amortization to algorithms for isolating the complex roots of a polynomial.

\subsection{Sturm Sequences}\label{sec:ComplexSturm}
In this section, we apply the continuous amortization technique to the Sturm-based algorithm for isolating the roots of a complex polynomial, as described in \cite{du-sharma-yap:sturm:07}.  In \cite{du-sharma-yap:sturm:07}, the authors discuss Sturm-based probes which determine the number of roots in half-open squares of the form $[a,b)\times(ci,di]$.  The correctness of this algorithm is based on the Routh-Hurwitz Theorem, forms of which are also contained in \cite{Routh,MR1416383,Pinkert:1976:EMF:355705.355710,Yap:Algebra,du-sharma-yap:sturm:07}

\begin{lemma}[Routh-Hurwitz Theorem \cite{Routh,MR1416383,Pinkert:1976:EMF:355705.355710,Yap:Algebra,du-sharma-yap:sturm:07}]  
Let $p$ be a complex valued polynomial of a real variable.  In other words, $p(x)=p_0(x)+p_1(x)i$ where $p_0(x)$ is the real part of $p$ and $p_1(x)i$ is the complex part of $p$.  Assume that $p_0$ and $p_1$ are relatively prime (in particular, neither is zero) and $\deg p_0\geq \deg p_1$.  Let $v$ be the variation on the real line of a generalized Sturm sequence whose first two polynomials are $p_1$ and $p_0$.  Then, the number of roots of $p$ lying above the real axis is $(\deg(p)-v)/2$.
\end{lemma}

This lemma can be applied to vertical or horizontal lines in the complex plane by considering $p(x+iy_0)$ or $p(x_0+iy)$ for the univariate polynomial $p$.  In this way, we count the number of roots to one side of an axis aligned line in the complex plane.  This technique can be extended to count the number of roots in a quadrant, which, in turn, can be used to count the number of roots in a square.  Note that the algorithm presented in \cite{du-sharma-yap:sturm:07} and suggested here requires many Sturm sequences to be computed because new sequences must be computed at every subdivision step.  These additional computations make the corresponding subdivision algorithm impractical, but the algorithm remains of theoretical interest.

The test described above can be used to count the number of roots in a square.  This test gives rise to inclusion and exclusion predicates in the following way.  A square $S$ is excluded if does not contain any roots, and $S$ is included if it contains exactly one root.  Therefore, our predicate $B_{\CSturm}$ for complex Sturm sequences is based on the number of roots in a square.
$$
B_{\CSturm}(S)=\begin{cases}\true&S\text{ contains 0 or 1 roots}\\
\false&S\text{ contains 2 or more roots}\end{cases}.
$$

We next derive a stopping function for the Sturm predicate.  Let the roots of $p$ be $\{\alpha_1,\cdots,\alpha_d\}$.  Following the example of Section \ref{sec:real:sturm}, we note that if a square $S$ containing $x$ also contains two roots $\alpha_{i_1}$ and $\alpha_{i_2}$ with $\alpha_{i_1}$ not further from $x$ than $\alpha_{i_2}$, i.e., $\dist(x,\alpha_{i_1})\leq\dist(x,\alpha_{i_2})$, then $\dist(x,\alpha_{i_2})=\dist_2(x,\{\alpha_i\})$.  Therefore, for $S$ to contain two roots, it is a necessary condition that $\diam(S)\geq\dist_2(x,\{\alpha_i\})$.  Thus, our stopping function $F_{\CSturm}$ restricts $S$ to have diameter less than $\dist_2(x,\{\alpha_i\})$.  Since the diameter of a square is $\sqrt{2}$ times its side length, a stopping function for the complex Sturm function is
$$
F_{\CSturm}=\left(\frac{1}{\sqrt{2}}\dist_2(x,\{\alpha_i\})\right)^2=\frac{1}{2}\dist_2(x,\{\alpha_i\})^2.
$$

In the case for real Sturm sequences, it is straightforward to determine the regions of integration based on the constructed stopping function, but for complex Sturm sequences, the regions of integration are much more complicated.  In particular, define the second-degree Voronoi cells as follows: the second-degree Voronoi cell for $\alpha_i$, denoted $R_{\alpha_i}$, is the set of points where $\alpha_i$ is the second closest root\footnote{The second-degree Voronoi cells are not the same as the second-order Voronoi cells \cite{MR1730176}, but are a union of cells in the common refinement of first- and second-order Voronoi cells.}.  In the region $R_{\alpha_i}$, the stopping function is given by $\dist(x,\alpha_i)/2$.  The regions $R_{\alpha_i}$ may be more complicated than standard Voronoi regions; in fact, they may not be convex and may be disconnected.

Consider the benchmark square $R=[-2^L,2^L]\times[-2^Li,2^Li]$, and let $Q_{\CSturm}$ be the partition of $R$ at the end of the complex Sturm bisection algorithm.  Using continuous amortization and the definition for the stopping function $F_{\CSturm}$, we can bound the size of the partition as follows:
\begin{equation}\label{equation:cSturm}
\#Q_{\texttt{CSturm}}(R)\leq \int_R{\frac{4dA}{\frac{1}{2}\dist_2(x,\{\alpha_i\})^2}}=\sum_{i=1}^d\int_{R_{\alpha_i}}{\frac{8dA}{|x-\alpha_i|^2}}.
\end{equation}

Since the regions $R_{\alpha_i}$ may be quite complicated, we construct larger regions which contain the $R_{\alpha_i}$, but are simpler to integrate over.  Define $d_i$ to be the distance from $\alpha_i$ to the nearest distinct root of $p$.  Then all of the points in the open disk of radius $d_i/2$ are closer to $\alpha_i$ than any other root.  Therefore, $R_{\alpha_i}$ cannot intersect this disk because for every point in $R_{\alpha_i}$, $\alpha_i$ is not the closest root.  In particular, let $D_{\alpha_i}$ consist of the annulus inside the closed disk of radius $\sqrt{2}\cdot 2^{L+1}$ centered at $\alpha_i$ and outside the open disk of radius $d_i/2$ centered at $\alpha_i$.  The outer radius is $\sqrt{2}\cdot 2^{L+1}$ because the diameter of $R$ is $\sqrt{2}\cdot 2^{L+1}$ and the disk of this radius contains all of $R$.  Since $R_{\alpha_i}$ is included in $R$ and $R_{\alpha_i}$ is excluded from the open disk of radius $d_i/2$ centered at $\alpha_i$, $R_{\alpha_i}\subseteq D_{\alpha_i}$.  Therefore, for each $i$, 
$$
\int_{R_{\alpha_i}}{\frac{8dA}{|x-\alpha_i|^2}}\leq\int_{D_{\alpha_i}}{\frac{8dA}{|x-\alpha_i|^2}}=\int_0^{2\pi}\int_{\frac{1}{2}d_i}^{\sqrt{2}\cdot 2^{L+1}}\frac{8r\,dr\,d\theta}{r^2},
$$
by converting to polar coordinates centered at $\alpha_i$.  Evaluating this integral, we have
$$
\int_{R_{\alpha_i}}{\frac{8dA}{|x-\alpha_i|^2}}\leq 16\pi\left(\ln\left(\sqrt{2}\cdot 2^{L+1}\right)-\ln\left(\frac{1}{2}d_i\right)\right).
$$
Substituting this inequality into Inequality \ref{equation:cSturm}, we have that 
\begin{equation}\label{equation:CSturmBound}
\#Q_{\texttt{CSturm}}(R)\leq\sum_{i=1}^d 16\pi\left(\ln\left(\sqrt{2}\cdot 2^{L+1}\right)-\ln\left(\frac{1}{2}d_i\right)\right).
\end{equation}
The terms whose logarithm has a positive coefficient are $O(L)$ and since there are $d$ roots in the sum, the part of the sum with positive coefficient is $O(dL)$.  In addition, the $\frac{1}{2}$ in the logarithm with negative coefficient can be separated out, and the terms of this type sum to $16\pi d\ln(2)$ which is $O(d)$.  The remainder can be bounded with the Mahler-Davenport bound of Lemma \ref{lemma:mahlerdavenport}.  In particular, we can construct a graph $G$ by connecting each root to one of its nearest neighbors.  This graph is an undirected multi-subgraph of the nearest neighbor graph \cite{EppsteinNearestNeighbor}.  It is a subgraph because each root is connected to only one of its nearest neighbors.  Since the maximum in-degree of the nearest neighbor graph is six \cite{EppsteinNearestNeighbor}, the maximum valence of $G$ is seven because a root can be connected to all six of its nearest neighbors; then there is one more nearest neighbor edge corresponding to the root itself.  Therefore, by Corollary \ref{corollary:MahlerDavenport}, the sum of the negative of logarithms of differences between roots in Inequality \ref{equation:CSturmBound} is $O(d(L+\ln d))$.  Combining all of these bounds shows that the size of the subdivision tree is $O(d(L+\ln d))$, matching the bounds in \cite{davenport:85,du-sharma-yap:sturm:07}.

\begin{remark}
In the computation above for complex Sturm sequences, we expanded the region of integration around each root of $p$, and, therefore, integrated over the same region many times.  In both of the real-root computations, such an approximation would have resulted in the same final complexity, but it seems more instructive to include the more precise calculation.
\end{remark}

\subsection{$\SqFreeCEVAL$}
In this section, we apply the continuous amortization technique to compute the size of the subdivision tree for the $\SqFreeCEVAL$ algorithm, which isolates the complex roots of a polynomial.  This algorithm is adapted from the 8-point test version of $\CEVAL$ \cite{SagraloffYap:8Point}, but it has one important distinction.  The $\CEVAL$ algorithm can be applied to any square-free polynomial, but in $\SqFreeCEVAL$, the additional condition is that $p'$ is assumed to also be square-free\footnote{The additional condition for $\SqFreeCEVAL$ arises because the application of the Mahler-Davenport bound requires that the roots of $pp'$ are distinct simple roots.  On the other hand, unlike in the case of \cite{BurrKrahmer:SqFreeEVAL} where this condition can be overcome, it is not immediately obvious how to apply the terminal conditions of $\CEVAL$ to the square-free component of $p'$.  In particular, there are counterexamples to the na\"ivest generalizations.}.  This restriction is mild, but does restrict the types of polynomials under consideration.

The inclusion and exclusion tests for $\SqFreeCEVAL$ are based on the following theorem
\begin{proposition}[\cite{SagraloffYap:8Point,SagraloffYap:CEVAL}]
Let $p$ be a square-free polynomial and $D$ a disk of radius $r$ centered at $m$.
\begin{enumerate}[(1)]
\item If $
\sum_{k=1}^d\left|\frac{p^{(k)}(m)}{k!\cdot p(m)}\right|r^k<1$, then $\overline{D}$ has no roots of $p$.
\item If $
\sum_{k=1}^{d-1}\left|\frac{p^{(k+1)}(m)}{k!\cdot p'(m)}\right|r^k<\frac{1}{\sqrt{2}}$, then $\overline{D}$ has at most one of $p$.\label{CEVAL:condition2}
\end{enumerate}
\end{proposition}
In Case \ref{CEVAL:condition2}, the disk $D$ is further tested using the 8-point test, see \cite{SagraloffYap:8Point}, to determine if $D$ actually contains a root.  The inclusion and exclusion tests for $\SqFreeCEVAL$ are based on these conditions as well as the 8-point test.  In particular, the predicate $B_{\SqFreeCEVAL}$ on squares for $\CEVAL$ follows:
\begin{subequations}\label{equation:SqFreeCEVAL}
\begin{equation}
B_{\SqFreeCEVAL}(S)=\begin{cases}
\true &\textstyle\sum_{k=1}^d\left|\frac{p^{(k)}(m(S))}{k!\cdot p(m(S))}\right|\left(\frac{\diam(S)}{2}\right)^k<1\\[4mm]
\true&\begin{split}\textstyle\sum_{k=1}^{d-1}&\textstyle\left|\frac{p^{(k+1)}(m(S))}{k!\cdot p'(m(S))}\right|(2\cdot\diam(S))^k<\frac{1}{6}\\
&\text{ and }\textstyle
\sum_{k=1}^{d-1}\left|\frac{p^{(k+1)}(m(S))}{k!\cdot p'(m(S))}\right|(4\cdot\diam(S))^k<\frac{1}{\sqrt{2}}\end{split}\\[4mm]
\false&\text{otherwise}\end{cases}.\tag{\theequation a,\theequation b,\theequation c}
\end{equation}
\end{subequations}
If Condition \ref{equation:SqFreeCEVAL}a is passed, then there are no roots of $p$ in $S$, and if Conditions \ref{equation:SqFreeCEVAL}b and \ref{equation:SqFreeCEVAL}c are passed, then the 8-point test must be applied.  Note that the two tests, Conditions \ref{equation:SqFreeCEVAL}b and \ref{equation:SqFreeCEVAL}c, could be combined into a single test using the largest left-hand-side and the smallest right-hand-side, but we will not need this simplification here.

Next, we derive a stopping function for the $\SqFreeCEVAL$ predicate.  For any univariate polynomial $f$, define 
$\Sigma_f(x)=\sum_{\alpha\in V(f)}\frac{1}{|x-\alpha|},$ where $V(f)$ is the set of roots of $f$.  In \cite{BurrKrahmerYap:ContinuousAmortization,BurrKrahmer:SqFreeEVAL}, it is shown that 
\begin{equation}\label{eq:harmonic}
\left|\frac{f^{n}(x)}{f(x)}\right|\leq\left(\Sigma_f(x)\right)^n.
\end{equation}
If $\diam(S)\leq \frac{1}{\Sigma_p(m(S))}$, then Condition \ref{equation:SqFreeCEVAL}a holds.  In particular, by substituting Inequality \ref{eq:harmonic} and the assumption into Condition \ref{equation:SqFreeCEVAL}a, we see that that sum is bounded above by $\sum_{k=1}^d\frac{1}{k!\cdot 2^k}\leq 1$.  Similarly, if $\diam(S)<\frac{1}{14\Sigma_{p'}(m(S)}$, then Condition \ref{equation:SqFreeCEVAL}a and \ref{equation:SqFreeCEVAL}b hold.  In particular, using this inequality and Inequality \ref{eq:harmonic} for $f'$ in Condition \ref{equation:SqFreeCEVAL}b gives us that the sum is bounded above by $\sum_{k=1}^{d-1}\frac{1}{k!\cdot 7^k}<\frac{1}{6}$ and Condition \ref{equation:SqFreeCEVAL}c is bounded above by $\sum_{k=1}^{d-1}\frac{2^k}{k!\cdot 7^k}<\frac{2}{5}<\frac{1}{\sqrt{2}}$.  The two inequalities discussed here can be used to derive stopping functions for the $\SqFreeCEVAL$ predicate.

In particular, we use the relationship between $\Sigma_p$ and harmonic means, as expressed in \cite{BurrKrahmer:SqFreeEVAL}; this relationship results in the following fact: if $\diam(S)<\frac{2}{3\cdot\Sigma_p(x)}$ or $\diam(S)<\frac{1}{21\cdot \Sigma_{p'}(x)}$ for $x\in S$, then $S$ is $\SqFreeCEVAL$ terminal.  Note that the better factor of $\frac{1+\ln 2}{2\ln 2}$ of \cite{SharmaYap:ContinuousAmortization} could also be used instead of the factor of $\frac{2}{3}$ used above, but that change would not affect our results significantly.  Next, since this is a two dimensional problem, we must convert these bounds into bounds on the area of $S$.  In other words, if $\area(S)<\frac{2}{9\cdot(\Sigma_p(x))^2}$ or $\area(S)<\frac{1}{882\cdot(\Sigma_{p'}(x))^2}$, then $S$ is $\SqFreeCEVAL$-terminal.  Therefore, these two functions are stopping functions for $\SqFreeCEVAL$.

The stopping functions computed above are not quite the functions that we will use because $\Sigma_p^2$ and $\Sigma_{p'}^2$ are sums of roots and involve mixed terms.  For any univariate polynomial $f$, define $\Sigma_f^2(x)=\sum_{\alpha\in V(f)}\frac{1}{|x-\alpha|^2}$.  This new function can be related to the stopping functions above using the Cauchy-Schwarz inequality.  In particular, $(\Sigma_f)^2\leq d\Sigma_f^2$.  Therefore, $\frac{2}{9d\cdot\Sigma^2_p(x)}$ and $\frac{1}{882d\cdot \Sigma^2_{p'}(x)}$ are stopping functions for $\SqFreeCEVAL$.

Finally, we determine the regions where the stopping functions are applied.  Let the roots of $pp'$ be $\{\alpha_1,\cdots,\alpha_{2d-1}\}$ where the first $d$ roots are the roots of $p$, and, for each $\alpha_i$, let $V_{\alpha_i}$ be the Voronoi cell containing $\alpha_i$.  For each $1\leq i\leq d$, the corresponding Voronoi cell $V_{\alpha_i}$ contains a root of $p$ and no roots of $p$.  Therefore, the function $\frac{1}{882d\cdot \Sigma^2_{p'}(x)}$ is positive in this region.  Similarly, for $d+1\leq i\leq 2d-1$, the Voronoi region $V_{\alpha_i}$ contains a root of $p'$ and no roots of $p'$.  Therefore, the function $\frac{2}{9d\cdot\Sigma^2_p(x)}$ is positive in this region.  Combining these facts, a stopping function for the $\SqFreeCEVAL$ predicate is
$$
F_{\SqFreeCEVAL}(x)=\begin{cases}
\frac{1}{882d\cdot \Sigma^2_{p'}(x)}&x\in V_{\alpha_i}\text{ with }1\leq i\leq d\\
\frac{2}{9d\cdot\Sigma^2_p(x)}&x\in V_{\alpha_i}\text{ with }d+1\leq i\leq 2d-1
\end{cases}.
$$

Next, we are ready to apply continuous amortization to compute the size of the subdivision tree over the benchmark square $R=[-2^L,2^L]\times[-2^Li,2^Li]$.  Let $Q_{\SqFreeCEVAL}(R)$ be the partition of $R$ at the end of the $\SqFreeCEVAL$ bisection algorithm.  Using continuous amortization with the function $F_{\SqFreeCEVAL}$, gives a bound on the size of the subdivision tree as
\begin{align*}
\#Q_{\SqFreeCEVAL}&(R)\leq\sum_{i=1}^d\int_{V_{\alpha_i}}3528d\cdot\Sigma^2_{p'}(x)dA+\sum_{i=d+1}^{2d-1}\int_{V_{\alpha_i}}18d\cdot\Sigma^2_p(x)dA\\
&\hspace{-.5in}=\sum_{i=1}^d\int_{V_{\alpha_i}}3528d\left(\sum_{j=d+1}^{2d-1}\frac{1}{|x-\alpha_j|^2}\right)dA
+
\sum_{i=d+1}^{2d-1}\int_{V_{\alpha_i}}18d\left(\sum_{j=1}^d\frac{1}{|x-\alpha_j|^2}\right)dA. 
\end{align*}
It is somewhat complicated to integrate over the Voronoi cells, so instead, we enlarge the regions of integration similarly to the computation in Section \ref{sec:ComplexSturm}.  In particular, let $d_i$ be the distance from $\alpha_i$ to the the nearest distinct root.  Then, an inscribed circle with center $\alpha_i$ in the Voronoi cell $V_{\alpha_i}$ has radius $d_i/2$.  Since $\frac{1}{|x-\alpha_i|^2}$ is never integrated over this disk of radius $d_i/2$, let $D_{\alpha_i}$ be the annulus between the disk of radius $\sqrt{2}\cdot 2^{L+1}$ centered at $\alpha_i$ and the disk of radius $d_i/2$ centered at $\alpha_i$.  Then this annulus contains $V_{\alpha_i}$ because the outer circle of the annulus contains the entire benchmark region $R$ and the Voronoi cell $V_{\alpha_i}$ does not include $D_{\alpha_i}$.  Therefore, the size of the subdivision tree can be bounded by
\begin{align*}
\#Q_{\SqFreeCEVAL}&(R)\leq \sum_{i=1}^d\int_{D_{\alpha_i}}\frac{18d}{|x-\alpha_i|^2}dA
+
 \sum_{i=d+1}^{2d-1}\int_{D_{\alpha_i}}\frac{3528d}{|x-\alpha_i|^2}dA\\
 &=
 \sum_{i=1}^d\int_0^{2\pi}\int_{\frac{1}{2}d_i}^{\sqrt{2}\cdot 2^{L+1}}\frac{18dr}{r^2}dr\,d\theta
+
 \sum_{i=d+1}^{2d-1}\int_0^{2\pi}\int_{\frac{1}{2}d_i}^{\sqrt{2}\cdot 2^{L+1}}\frac{3528dr}{r^2}dr\,d\theta.
\end{align*}
Therefore, 
\begin{multline*}
\#Q_{\SqFreeCEVAL}(R)\\\leq\sum_{i=1}^d 36d\pi\left(\ln\left(\sqrt{2}\cdot 2^{L+1}\right)-\ln\left(\frac{1}{2}d_i\right)\right)
+
\sum_{i=d+1}^{2d-1} 7056d\pi\left(\ln\left(\sqrt{2}\cdot 2^{L+1}\right)-\ln\left(\frac{1}{2}d_i\right)\right).
\end{multline*}

Since there are $2d-1$ roots of $pp'$ and each term whose logarithm has a positive coefficient is $O(dL)$, the part of the sum whose logarithm has positive coefficient is $O(d^2L)$.  In addition, the $\frac{1}{2}$ in the logarithm with negative coefficient can be separated out, and since there are $2d-1$ terms of this type, this part of the sum contributes $O(d^2)$.  The final part of the sum can be bounded with the Mahler-Davenport bound of Lemma \ref{lemma:mahlerdavenport}.  In particular, we construct a graph $G$ by connecting each root to one of its nearest neighbors.  As in Section \ref{sec:ComplexSturm}, this graph is an undirected multi-subgraph of the nearest neighbor graph and its valence is bounded by seven.  Note that $M(pp')=M(p)M(p')$, which is bounded by $\|p\|_2\|p'\|_2$.  By assumption, $\|p\|_2\leq\sqrt{d+1}\cdot 2^L$ and thus $\|p'\|\leq d\sqrt{d}\cdot 2^L$.  Therefore, $\ln(M(pp'))$ is $O(d(L+\ln d))$.  By Corollary \ref{corollary:MahlerDavenport} (after substituting this bound for the Mahler measure) gives that the total complexity of the subdivision tree is $O(d^2(L+\ln d))$.  This bound improves upon the bound in \cite{SagraloffYap:8Point} by a factor of $(\ln d)^2$ in the case where $p'$ is square-free.

\section{Bit Complexity}\label{sec:bit}
In this section, we extend the continuous amortization technique to compute the bit-complexity of bisection-based techniques.  In particular, we present a generalization of continuous amortization to compute the value of a function on the leaves of a subdivision tree.  We use this technique to compute the bit-complexity for Descartes' rule of signs \cite{eigenwillig-sharma-yap:descartes:06} and $\SqFreeEVAL$ \cite{BurrKrahmer:SqFreeEVAL}.  Our complexity bounds match or surpass the best bounds in the literature.

\subsection{Generalizing Continuous Amortization to Bit Complexity}
In this section, we consider the more general problem of evaluating a function $g$ on each of the terminal intervals of bisection-based algorithms on $\mathbb{R}$.  In particular, let $g$ be a positive and decreasing real-valued function (whose constants may depend on the problem data).  We think of $g$ as a function on intervals where smaller intervals are more costly.  This section shows how to bound the value of the sum of the values of $g$ on the terminal intervals of Algorithm \ref{algorithm:general}.  These conditions on $g$ are not very restraining, in practice, and they are easily satisfied by the examples in this section.  This setup allows us to extend continuous amortization to compute a bound on the sum of $g$ applied to the terminal intervals of Algorithm \ref{algorithm:general}.  In particular,

\begin{proposition}\label{prop:bitcomplex}
Let $I$ be an initial interval and let $Q(I)$ be the partition of $I$ at the end of Algorithm \ref{algorithm:general}.  Let $g$ be a positive and decreasing function (whose constants may depend on the problem data).  In addition, let $F$ be a stopping function for Algorithm \ref{algorithm:general}.  Then,
$$
\sum_{J\in Q(I)}g(w(J))\leq\max\left\{g(w(I)),\int_I\frac{2g(F(x)/2)}{F(x)}dx\right\}.
$$
If the algorithm does not terminate, then the integral is infinite.
\end{proposition}
\begin{proof}
As in Proposition \ref{prop:ndimensional}, this proof follows the argument in \cite{BurrKrahmer:SqFreeEVAL}.  If $\#Q(I)=1$, then the bound is immediate.  If $\#Q(I)>1$, then let $J$ be any interval in the partition $Q(I)$.  There is a lower bound on the $w(J)$ because Algorithm \ref{algorithm:general} did not terminate on the parent of $J$.  In particular, 
$$
\forall y\in J, w(J)\geq\frac{1}{2}F(y).
$$
On the other hand, since 
$$
\int_I\frac{2g(F(x)/2)}{F(x)}dx=\sum_{J\in Q(I)}\int_J\frac{2g(F(x)/2)}{F(x)}dx,
$$
it is sufficient to show that the value of this integral on each terminal interval $J$ is at least $g(w(J))$.  For a fixed $J$, let $z\in S$ be such that $F(z)$ is maximal in $J$.  Then,
\begin{equation*}
\int_J\frac{2g(F(x)/2)}{F(x)}dx\geq\int_J\frac{2g(F(z)/2)}{F(z)}dx\geq \frac{2g(F(z)/2)}{F(z)}\cdot w(J)
\geq \frac{2g(w(J))}{F(z)}\cdot \frac{F(z)}{2}=g(w(J)).
\end{equation*}
If Algorithm \ref{algorithm:general} does not terminate then the integral above is a bound for the value of $g$ at any point in time, but since $g$ is positive and increases with more subdivisions, this integral is greater than an arbitrarily large sum of positive values which are bounded below, and, hence infinite.
\end{proof}

\begin{remark}\label{remark:depth}
In most of the cases we will be considering, the function $g$ can be written in terms of the depth of the tree.  We can understand such a function as follows: Let $I$ be the initial interval for a bisection algorithm; each time $I$ is subdivided, the width of the resulting intervals are decreased by half.  Therefore, for an interval $J$ in the subdivision tree for Algorithm \ref{algorithm:general}, the depth of $J$ is $\log_2(w(I)/w(J))$ (where the depth of the root is assumed to be 0).  Note also that this function is decreasing in $w(J)$.  From the proof of Proposition \ref{prop:bitcomplex}, we see that $w(J)\geq F(x)/2$.  Therefore, in most of our applications, the integrand will include a term of the form $\log_2(2w(I)/F(x))$, which is an upper bound on the depth of the subdivision tree at $x\in I$.
\end{remark}

\begin{remark}
Note that the standard continuous amortization construction is the special case of this extension where the function $g\equiv1$.
\end{remark}

\begin{remark}
This extension of continuous amortization to compute the complexity of the sum of a function applied to terminal intervals can be directly combined with the extension of continuous amortization to higher dimensions of Section \ref{sec:CAHigher} in an analogous way to the extension of one-dimensional continuous amortization to higher dimensions.
\end{remark}

\begin{remark}
Note that by using H\"older's inequality (with $p=1$ and $q=\infty$) on this integral we have the following product.
\begin{subequations}\label{equation:product}
\begin{equation*}
\int_I\frac{2g(F(x)/2)}{F(x)}dx\leq\left(\int_I\frac{2dx}{F(x)}\right)\left(\sup_{x\in I} g(F(x)/2)\right)\eqlabel
\end{equation*}
\end{subequations}
Factor \ref{equation:product}a is the standard continuous amortization bound for the size of the subdivision tree, and Factor \ref{equation:product}b is an upper bound for $g$ at the deepest point of the subdivision tree (i.e., a global upper bound on the value of $g$ on any node in the tree).  This product corresponds to one of the standard techniques for computing the bit-complexity of an algorithm.  In particular, the maximum bit cost is multiplied by the size of the tree.  The value of the continuous amortization integral, however, may result in a tighter bound than what would be computed through this standard technique, see Remark \ref{remark:sqfreeevalbit} for such a discussion.
\end{remark}

\begin{remark}
In the next two sections, we apply this technique for computing the bit-complexity of two algorithms for the isolation of the real roots of a polynomial.  In these cases, it is known that the bit complexity of the algorithm is $\widetilde{O}(d^4L^2)$ and that the bit complexity of a single node includes a term of the form $\widetilde{O}(d^2h)$.  It is also known that the depth of the tree may be $dL$.  A simple lower bound can be found by adding up the bit-complexity cost along a path of maximum length.  In this computation, the resulting complexity is $\widetilde{O}(d^4L^2)$; note that this is a bound for a single path in the tree and a lower bound for the bit-cost of the entire tree.  With tight bound, such as these, we cannot hope for a asymptotic improvement in the leading terms, but the computations using continuous amortization presented in the next sections result in better bounds in special cases.
\end{remark}

\subsection{Descartes' Rule of Signs}\label{sec:Descartesbit}
In this section, we apply the continuous amortization technique for computing the bit-complexity of a function to the Descartes' rule of signs algorithm.  In \cite{eigenwillig-sharma-yap:descartes:06}, it is shown that bit-complexity of computing at a node of depth $h$ in the subdivision tree is $O(d^3L+d^3h)$ using classical arithmetic and $\widetilde{O}(d^2L+d^2h)$ using asymptotically fast multiplication.  In this section, we use the stopping function $F_{\Descartes}$, as derived in Section \ref{sec:Descartes}.  Finally, since the benchmark interval is $I=[-2^L,2^L]$, the benchmark interval has width $w(I)=2^{L+1}$.

We cannot apply continuous amortization to this problem directly because the bit-complexity costs quoted above are charged to both internal nodes and leaves of the subdivision tree.  In order to turn the bit-complexity into a function on the leaves of the tree, we divide the cost of each internal node among its two children.  This cost accumulates in the leaves of the tree.  In particular, for a leaf of depth $h_0$, the cost accumulated in the leaf is 
$$
\sum_{i=0}^{h_0-1}\frac{1}{2^i}C(d^3L+d^3(h_0-i))\leq 2C(d^3L+d^3h_0)
$$
for classical arithmetic, where $C$ is the constant suppressed by the $O$.  In addition, in the asymptotically fast case, the cost in the leaf of depth $h_0$ is
$$
\sum_{i=0}^{h_0-1}\frac{1}{2^i}(d^2L+d^2(h_0-i))k(d,L,h_0-i)\leq 2(d^2L+d^2h_0)k(d,L,h_0),
$$
where $k(d,L,h)$ holds the leading constant and logarithmic factors for a node of depth $h$, which is suppressed by the $\widetilde{O}$.  These remaining functions are functions on the leaves of the tree and continuous amortization applies to them.

Beginning classically, we evaluate the following function, where the expression of Remark \ref{remark:depth} has been substituted for $h$.
\begin{subequations}\label{equation:Descartes:bit}
\begin{multline*}
\int_I\frac{4C(d^3L+d^3\log_2(2w(I)/F_{\Descartes}(x)))}{F_{\Descartes(x)}}dx\\=2C(2d^3L+2d^3)\int_I\frac{2dx}{F_{\Descartes}(x)}-2Cd^3\int_I\frac{2\log_2(F_{\Descartes}(x))}{F_{\Descartes}(x)}dx.\eqlabel
\end{multline*}
\end{subequations}
Integral \ref{equation:Descartes:bit} is just the standard continuous amortization bound on the size of the subdivision tree.  Integral \ref{equation:Descartes:bit} is calculated below.  Similarly, for asymptotically fast multiplication, continuous amortization requires the computation of the following integral
\begin{equation}\label{Descartes:bit:asymptotic}
\int_I\frac{4(d^2L+d^2\log_2(2w(I)/F_{\Descartes}(x)))k(d,L,h)}{F_{\Descartes(x)}}dx.
\end{equation}
Using H\"older's inequality, the logarithmic factor $k(d,L,h)$ can be bounded by its value on the maximum depth of the tree, which is $O(d(L+\ln d))$.  Therefore, $k$ contributes a constant times logarithmic factors in $d$ and $L$ and will be ignored.  The remaining portion of Integral \ref{Descartes:bit:asymptotic} can be bounded as:
\begin{subequations}\label{equation:Descartes:bit:asymptotic}
\begin{multline*}
\int_I\frac{4(d^2L+d^2\log_2(2w(I)/F_{\Descartes}(x)))}{F_{\Descartes(x)}}dx\\
=2(2d^2L+2d^2)\int_I\frac{2dx}{F_{\Descartes}(x)}-2d^2\int_I\frac{2\log_2(F_{\Descartes}(x))}{F_{\Descartes}(x)}dx.\eqlabel
\end{multline*}
\end{subequations}
Note that the integrals in Equation \ref{equation:Descartes:bit:asymptotic} are exactly the same as in Equation \ref{equation:Descartes:bit}.  In particular, we have already noted that Integral \ref{equation:Descartes:bit:asymptotic}a computes a bound on the size of the subdivision tree.  From Section \ref{sec:Descartes}, the subdivision tree is $O(d(L+\ln d))$.  Substituting this bound into Integrals \ref{equation:Descartes:bit}a and \ref{equation:Descartes:bit:asymptotic}a results in complexities of $\widetilde{O}(d^4L^2)$ and $\widetilde{O}(d^3L^2)$, respectively.  These will be seen to be lower-order terms than the contributions of Integrals \ref{equation:Descartes:bit}b and \ref{equation:Descartes:bit:asymptotic}b.  It remains to compute Integrals \ref{equation:Descartes:bit}b and \ref{equation:Descartes:bit:asymptotic}b, which is the same integral, i.e., 
$$
-\int_I\frac{2\log_2(F_{\Descartes}(x))}{F_{\Descartes}(x)}dx,
$$
which we now evaluate.  The first step is to convert the $\log_2$ to a natural logarithm by dividing by $\ln 2$.  Let $\{\alpha_1,\cdots,\alpha_d\}$ be the roots of $p$, where, w.l.o.g., the first $k$ roots are the real roots.  Recall that for real roots, we use both $F_{\Descartes,1}$ and $F_{\Descartes,2}$ as stopping functions.  Let $\alpha_j$ be a real root; then, we calculate
\begin{subequations}\label{bit:Descartes:parts}
\begin{multline*}
-\int_{I_{\alpha_j}}\frac{2\log_2(F_{\Descartes}(x))}{F_{\Descartes}(x)}dx=\\
-\int_{I_{\alpha_j}\setminus J_{\alpha_j}}{\frac{2\ln(F_{\Descartes,1}(x))}{F_{\Descartes,1}(x)}}
dx
-\int_{\alpha_j-\frac{1}{1+\sqrt{3}}d_j}^{\alpha_j+\frac{1}{1+\sqrt{3}}d_j}{\frac{2\ln(F_{\Descartes,2}(x))}{F_{\Descartes,2}(x)}}
dx\eqlabel
\end{multline*}
\end{subequations}

For Integral \ref{bit:Descartes:parts}a, we evaluate it as follows
\begin{multline*}
-\int_{I_{\alpha_j}\setminus J_{\alpha_j}}{\frac{2\ln(F_{\Descartes,1}(x))}{F_{\Descartes,1}(x)}}
dx=-\int_{a_j}^{\alpha_j-\frac{1}{1+\sqrt{3}}d_j}{\frac{2\ln(\alpha_j-x)}{\alpha_j-x}}dx-\int_{\alpha_j+\frac{1}{1+\sqrt{3}}d_j}^{b_j}{\frac{2\ln(x-\alpha_j)}{x-\alpha_j}}dx\\
=2\left(\ln\left(\frac{1}{1+\sqrt{3}}d_j\right)\right)^2-\left(\ln(\alpha_j-a_j)\right)^2-\left(\ln(b_j-\alpha_j)\right)^2.
\end{multline*}
The terms with negative coefficient are bounded above by zero and can be ignored.  The term with positive coefficient can be bounded above by 
$$2\left(\ln d_j\right)^2-4\ln\left(1+\sqrt{3}\right)\ln d_j+2\left(\ln\left(1+\sqrt{3}\right)\right)^2.$$  The term $2\left(\ln\left(1+\sqrt{3}\right)\right)^2$
appears once for each root, and, since there are $d$ roots, is $O(d)$.  The term $-4\ln\left(1+\sqrt{3}\right)\ln d_j$ becomes the negative of a sum of logarithms of distances between roots.  This sum was proved to be $O(d(L+\ln d))$ in Section \ref{sec:Descartes}.  The remaining terms could be bounded by $O\left(d^2(L+\ln d)^2\right)$ using the Mahler-Davenport bound, Lemma \ref{lemma:mahlerdavenport}, but, for now, we do not evaluate that term.

To bound Integral \ref{bit:Descartes:parts}b, we use, as above, the Inequality \ref{eq:Descartes}.  This integral is then
\begin{multline*}
-\int_{\alpha_j-\frac{1}{1+\sqrt{3}}d_j}^{\alpha_j+\frac{1}{1+\sqrt{3}}d_j}{\frac{2\ln(F_{\Descartes,2}(x))}{F_{\Descartes,2}(x)}}
dx\leq-\int_{\alpha_j-\frac{1}{1+\sqrt{3}}d_j}^{\alpha_j+\frac{1}{1+\sqrt{3}}d_j}{\frac{2\sqrt{3}\ln\left(\frac{1}{\sqrt{3}}(d_j-|\alpha_j-x|)\right)}{d_j-|\alpha_j-x|}}dx\\
=4\sqrt{3}\ln\left(\frac{\sqrt{3}}{1+\sqrt{3}}\right)\ln d_j+2\sqrt{3}\ln\left(\frac{1+\sqrt{3}}{\sqrt{3}}\right)\ln(3+\sqrt{3}).
\end{multline*}
Since the roots are all bounded in magnitude by $2^L$ and there are at most $d$ terms of the above form, one for each root, these integrals are $\widetilde{O}(dL)$.  In the final computation, this term is a lower-order term and can be ignored.

The final integral to evaluate occurs when $\alpha_j$ is a complex root.  In this case, only $F_{\Descartes,1}$ applies as follows:  
\begin{subequations}\label{Descartes:bit:complex}
\begin{multline*}
-\int_{I_{\alpha_j}}{\frac{2\ln(F_{\Descartes,1}(x))}{F_{\Descartes}(x)}}dx=-\int_{I_{\alpha_j}}{\frac{2\ln(F_{\Descartes,1}(x))}{F_{\Descartes,1}(x)}}dx\\
=-\int_{a_j}^{b_j}\frac{2\ln(|\alpha_j-x|)}{|\alpha_j-x|}dx
=-\int_{a_j}^{b_j}\frac{2\ln(2|\alpha_j-x|)-2\ln 2}{|\alpha_j-x|}dx\\
\leq -\int_{a_j}^{b_j}\frac{2\ln(x-\real(\alpha_j)+|\alpha_j-x|)}{|\alpha_j-x|}dx+\int_{a_j}^{b_j}\frac{2\ln 2}{|\alpha_j-x|}dx.\eqlabel
\end{multline*}
\end{subequations}
The inequality holds because $|x-\real(\alpha_j)|\leq |\alpha_j-x|$ and the sign of the integral is negative.  Integral \ref{Descartes:bit:complex}a was calculated in Section \ref{sec:Descartes} (up to constant factors) and is $\widetilde{O}(dL)$.  In the final computation, this term is a lower-order term and can be ignored.  The derivative of the logarithm in Integral \ref{Descartes:bit:complex}b is, up to sign, the denominator of this integral.  Therefore, Integral \ref{Descartes:bit:complex}b simplifies to
\begin{multline*}
\shoveleft{-\int_{a_j}^{b_j}\frac{2\ln(x-\real(\alpha_j)+|\alpha_j-x|)}{|\alpha_j-x|}dx\hfill}\\
\shoveright{\hfill=2\left(\ln(|\imag(\alpha_j)|)\right)^2-\left(\ln\left(\real(\alpha_j)-a_j+|\alpha_j-a_j|\right)\right)^2-\left(\ln\left(b_j-\real(\alpha_j)-|\alpha_j-b_j|\right)\right)^2.}
\end{multline*}
The terms with negative coefficient are bounded above by zero and can be ignored.  Collecting all of the highest order terms shows that they are $\sum 2\left(\ln d_j\right)^2$ for real roots and $\sum 2\left(\ln(|\imag(\alpha_j)|)\right)^2$ for complex roots.  These are the same sums, but squared, as were computed for the size of the subdivision tree for Descartes' rule of signs, see Section \ref{sec:Descartes}.  Since this sum is less than the square of the sum of the logarithms, these sums are $O(d^2(L+\ln d)^2)$.  Combining all of the data above, the resulting complexities are $O(d^5(L+\ln d)^2)$ for classical arithmetic and $\widetilde{O}(d^4L^2)$ for asymptotically fast multiplication, matching the bound in \cite{eigenwillig-sharma-yap:descartes:06}.

\begin{remark}
One may ask how the constants in these bounds compare with those in \cite{eigenwillig-sharma-yap:descartes:06}.  The coefficient of the sum of squares of logarithms (base 2) in the calculation above is 4 while the coefficient in \cite{eigenwillig-sharma-yap:descartes:06} can be made to be 1.  To achieve this coefficient of 1, some of the ideas from the above calculation must be incorporated into \cite{eigenwillig-sharma-yap:descartes:06}.  In particular, without going into the details of the calculation in \cite{eigenwillig-sharma-yap:descartes:06}, the charging scheme which charges the bit-complexity to the leaves must be tuned for the trimmed tree $T'$ in \cite{eigenwillig-sharma-yap:descartes:06}.  In this case, the cost for each leaf of the trimmed tree is twice the sum of the bit-costs of the path from the leaf to the root (the factor of 2 accounts for trimmed leaves along the path).

The main reason that this computation results in an extra factor of 4 is that the stopping function assumes that each root can influence neighbors to both its left and right; whereas, in \cite{eigenwillig-sharma-yap:descartes:06}, it is shown that each terminal interval in the trimmed tree $T'$ can be associated with a unique root and that unique root lies to one side of the interval (not both sides).  Because our stopping function does not reflect this fact, the results for Descartes' rule of signs are two (or two-squared) times too large.
\end{remark}

\subsection{\texttt{SqFreeEVAL}}
The $\SqFreeEVAL$ algorithm was introduced in \cite{BurrKrahmer:SqFreeEVAL} and other versions of it were studied in \cite{BurrKrahmerYap:ContinuousAmortization,SagraloffYap:8Point,SagraloffYap:CEVAL,SharmaYap:ContinuousAmortization}.  The algorithm $\SqFreeEVAL$ takes, as input, the square free parts of polynomial\footnote{Throughout this section, we assume that $p$ and $p'$ are both square free.  If $p$ or $p'$ were replaced by its square-free components, then small changes must be made to the text, especially when discussing the degrees of these polynomials.  The final results, however, do not change.} $p$ and its derivative $p'$ and isolates the roots of $p$.  In \cite{BurrKrahmer:SqFreeEVAL}, continuous amortization was used to show that the complexity of this algorithm is $O(d(L+\ln d))$.  We briefly recall the pertinent information for this computation.

The predicate for $\SqFreeEVAL$ is the following function on intervals $J$
$$
B_{\SqFreeEVAL}(J)=\begin{cases}\true&|p(m(J))|>\sum_{i=1}^d\frac{|p^{(i)}(m(J))|}{i!}\left(\frac{w(J)}{2}\right)^i\\
\true &|p'(m(J))|>\sum_{i=1}^{d-1}\frac{|p^{(i+1)}(m(J))|}{i!}\left(\frac{w(J)}{2}\right)^i\\
\false&\text{otherwise}
\end{cases}.
$$
These inequalities come from applying a reverse triangle inequality to the Taylor series for $p$ and $p'$ centered at $m(J)$.  The exact formulation of $B_{\SqFreeEVAL}$ is not important for this section because we use the stopping functions for $B_{\SqFreeEVAL}$ which were developed in \cite{BurrKrahmer:SqFreeEVAL} for this predicate.  In particular, let $\{\alpha_1,\cdots,\alpha_{2d-1}\}$ be the combined roots of $p$ and $p'$, and, assume, w.l.o.g., that $\{\alpha_1,\cdots,\alpha_d\}$ are the roots of $p$.  For each root $\alpha_i$, let $I_{\alpha_i}$ be the set of points in the benchmark interval $I=[-2^L,2^L]$ whose distance to the root $\alpha_i$ is not more than the distance to the other roots of $p$ and $p'$.  Then, a stopping function for this predicate is based on two stopping functions:
\begin{subequations}\label{equation:SqFreeEVAL:stopping}
\begin{multline*}
F_{\SqFreeEVAL}(x)=\\
\begin{cases}F_{\SqFreeEVAL,1}(x)=\frac{2}{3}\left/\left(\sum_{j=d+1}^{2d-1}\frac{1}{|x-\alpha_j|}\right)\right.&x\in I_{\alpha_i}\text{ with }1\leq i\leq d\\
F_{\SqFreeEVAL,2}(x)=\frac{2}{3}\left/\left(\sum_{j=1}^d\frac{1}{|x-\alpha_j|}\right)\right.&x\in I_{\alpha_i}\text{ with }d+1\leq i\leq 2d-1
\end{cases}.\eqlabel
\end{multline*}
\end{subequations}
In particular, Expression \ref{equation:SqFreeEVAL:stopping}a is a sum over the roots of $p'$ and is used near roots of $p$; while Expression \ref{equation:SqFreeEVAL:stopping}b is a sum over the roots of $p$ and is used near roots of $p'$.  Additionally, in \cite{SagraloffYap:8Point,SagraloffYap:CEVAL}, it is shown that the bit-complexity at a node of the subdivision tree of depth $h$ is $\widetilde{O}(dL+d^2h)$, using asymptotically fast Taylor shifts.  We will write the complexity for a node as $(dL+d^2h)k(d,L,h)$ where $k$ includes the coefficients and logarithmic terms.

Using the same tricks as in Section \ref{sec:Descartesbit}, we charge the leaves for the bit complexity of the internal nodes by splitting each internal node's cost in half among its two children.  From this computation, we know that the bit-cost for a leaf of depth $h$ is bounded by $2(dL+d^2h)k(d,L,h)$.  Now, using this formulation and continuous amortization for bit-complexity, the bit complexity of the $\SqFreeEVAL$ algorithm is
$$
\int_I\frac{4(dL+d^2\log_2(2w(I)/F_{\SqFreeEVAL}(x)))k(d,L,h)}{F_{\SqFreeEVAL}(x)}dx.
$$
Using H\"older's inequality, the logarithmic factor $k(d,L,h)$ can be factored out and bounded in terms of the maximum depth of the subdivision tree, which is $O(d(L+\ln d))$.  Therefore, $k$ contributes a constant times logarithmic factors in $d$ and $L$ and will be ignored.  The remainder of the integral can be bounded as follows:
\begin{subequations}\label{equation:sqfreeeval:split}
\begin{multline*}
\int_I\frac{4(dL+d^2\log_2(2w(I)/F_{\SqFreeEVAL}(x)))}{F_{\SqFreeEVAL}(x)}dx\\=2(dL+d^2L+2d^2)\int_I\frac{2dx}{F_{\SqFreeEVAL}(x)}+2d^2\int_I\frac{2\log_2(1/F_{\SqFreeEVAL}(x))}{F_{\SqFreeEVAL}(x)}dx.\eqlabel
\end{multline*}
\end{subequations}
Integral \ref{equation:sqfreeeval:split}a is a bound on the size of the subdivision tree using continuous amortization, which was computed in \cite{BurrKrahmer:SqFreeEVAL} to be $O(d(L+\ln d))$.  Therefore, Integral \ref{equation:sqfreeeval:split}a is bounded above by $\widetilde{O}(d^3L)$.  It will turn out that this is a lower-order term and can be ignored.

Next, we evaluate the Integral \ref{equation:sqfreeeval:split}b.  The first step is to convert the $\log_2$ to a natural logarithm by dividing by $\ln 2$.  Then Integral \ref{equation:sqfreeeval:split}b can be split along the roots $\alpha_i$.  In particular, 
\begin{subequations}\label{equation:sqfreeeval:second}
\begin{multline*}
\int_I\frac{2\ln(1/F_{\SqFreeEVAL}(x))}{F_{\SqFreeEVAL}(x)}dx\\
=\sum_{i=1}^d\int_{I_{\alpha_i}}\frac{2\ln(1/F_{\SqFreeEVAL,1}(x))}{F_{\SqFreeEVAL,1}(x)}dx+\sum_{i=d+1}^{2d-1}\int_{I_{\alpha_i}}\frac{2\ln(1/F_{\SqFreeEVAL,2}(x))}{F_{\SqFreeEVAL,2}(x)}dx.\eqlabel
\end{multline*}
\end{subequations}
Both $1/F_{\SqFreeEVAL,1}$ and $1/F_{\SqFreeEVAL,2}$ are composed of sums of reciprocals of distances to roots.  Before calculating these integrals, we use the log-sum inequality, see, e.g., \cite{Csiszar:2004:ITS:1166379.1166380}, to make these integrals easier to calculate.  In particular for Integral \ref{equation:sqfreeeval:split}a,
\begin{align}
\sum_{i=1}^d\int_{I_{\alpha_i}}\frac{2\ln(1/F_{\SqFreeEVAL,1}(x))}{F_{\SqFreeEVAL,1}(x)}dx
&=
\sum_{i=1}^d\int_{I_{\alpha_i}}3\left(\sum_{j=d+1}^{2d-1}\frac{1}{|x-\alpha_j|}\right)\ln\left(\frac{3}{2}\sum_{j=d+1}^{2d-1}\frac{1}{|x-\alpha_j|}\right)dx\notag\\
&\leq\sum_{i=1}^d\int_{I_{\alpha_i}}\sum_{j=d+1}^{2d-1}\frac{3}{|x-\alpha_j|}\ln\left(\frac{3d}{2|x-\alpha_j|}\right)dx.\label{equation:sqfree:eval:afterlogsum}
\end{align}
This sum of integrals skips over all intervals $I_{\alpha_i}$ for the roots of $p'$; we, however, can enlarge the region of integration, where each root of $p'$ is integrated over the larger interval containing all but its corresponding interval.  In particular,
$$
\sum_{i=1}^d\int_{I_{\alpha_i}}\sum_{j=d+1}^{2d-1}\frac{3}{|x-\alpha_j|}\ln\left(\frac{3d}{2|x-\alpha_j|}\right)dx
\leq \sum_{j=d+1}^{2d-1}\int_{I\setminus I_{\alpha_j}}\frac{3}{|x-\alpha_j|}\ln\left(\frac{3d}{2|x-\alpha_j|}\right)dx.
$$
A similar computation can be performed on Integral \ref{equation:sqfreeeval:split}b to find that
\begin{align*}
\sum_{i=d+1}^{2d-1}\int_{I_{\alpha_i}}\frac{2\ln(1/F_{\SqFreeEVAL,2}(x))}{F_{\SqFreeEVAL,2}(x)}dx
&=
\sum_{i=d+1}^{2d-1}\int_{I_{\alpha_i}}3\left(\sum_{j=1}^{d}\frac{1}{|x-\alpha_j|}\right)\ln\left(\frac{3}{2}\sum_{j=1}^{d}\frac{1}{|x-\alpha_j|}\right)dx\\
&\leq \sum_{j=1}^{d}\int_{I\setminus I_{\alpha_j}}\frac{3}{|x-\alpha_j|}\ln\left(\frac{3d}{2|x-\alpha_j|}\right)dx.
\end{align*}
Therefore, Integral \ref{equation:sqfreeeval:second} can be bounded with 
$$
\int_I\frac{2\ln(1/F_{\SqFreeEVAL}(x))}{F_{\SqFreeEVAL}(x)}dx\leq \sum_{j=1}^{2d-1}\int_{I\setminus I_{\alpha_j}}\frac{3}{|x-\alpha_j|}\ln\left(\frac{3d}{2|x-\alpha_j|}\right)dx.
$$
Separating out the $3d/2$ in the logarithm results in
\begin{subequations}\label{sqfreeEVAL:final}
\begin{multline*}
\sum_{j=1}^{2d-1}\int_{I\setminus I_{\alpha_j}}\frac{3}{|x-\alpha_j|}\ln\left(\frac{3d}{2|x-\alpha_j|}\right)dx\\
=
\sum_{j=1}^{2d-1}\int_{I\setminus I_{\alpha_j}}\frac{3}{|x-\alpha_j|}\ln\left(\frac{3d}{2}\right)dx
-
\sum_{j=1}^{2d-1}\int_{I\setminus I_{\alpha_j}}\frac{3}{|x-\alpha_j|}\ln\left(|x-\alpha_j|\right)dx.\eqlabel
\end{multline*}
\end{subequations}
Integral \ref{sqfreeEVAL:final}a, without the $\ln d$ factor, was already computed in \cite{BurrKrahmer:SqFreeEVAL} to be $O(d(L+\ln d))$.  Integral \ref{sqfreeEVAL:final}a, therefore, is $\widetilde{O}(dL)$.  In the final computation, this is a lower-order term and can be ignored.  

The remaining integral, Integral \ref{sqfreeEVAL:final}b, can be evaluated as follows.  Let $\alpha_j$ be real and let $I_{\alpha_j}=[a_j,b_j]$.  Then,
\begin{multline*}
-\int_{I\setminus I_{\alpha_j}}\frac{3}{|x-\alpha_j|}\ln\left(|x-\alpha_j|\right)dx=-3\int_{-2^L}^{a_j}\frac{\ln(\alpha_j-x))}{\alpha_j-x}dx-3\int_{b_j}^{2^L}\frac{\ln(x-\alpha_j)}{x-\alpha_j}dx\\
=\frac{3}{2}(\ln(|\alpha_j-a_j|))^2+\frac{3}{2}(\ln(|\alpha_j-b_j|))^2-\frac{3}{2}(\ln(\alpha_j+2^L))^2-\frac{3}{2}(\ln(2^L-\alpha_j))^2.
\end{multline*}
The terms with negative leading coefficient are bounded above by zero.  The terms with positive coefficient can be bounded in terms of the logarithm of distances between roots.  These terms are squares of terms that appear in the analysis in \cite{BurrKrahmer:SqFreeEVAL}.  There, the negative of a sum of logarithms are shown to be $O(d(L+\ln d))$; therefore, this sum can be bounded by $O(d^2(L+\ln d)^2)$.

Finally, let $\alpha_j$ be a complex root, we begin by enlarging the region of integration in Integral \ref{sqfreeEVAL:final}b from $I\setminus I_{\alpha_i}$ to $I$, and bound Integral \ref{sqfreeEVAL:final}b as in Section \ref{sec:Descartesbit}.  In particular,
\begin{subequations}\label{equation:sqfreeeval:complex}
\begin{multline*}
-\int_{I\setminus I_{\alpha_j}}\frac{3}{|x-\alpha_j|}\ln\left(|x-\alpha_j|\right)dx\leq
-\int_{I}\frac{3}{|x-\alpha_j|}\ln\left(|x-\alpha_j|\right)dx\leq\int_I\frac{3\ln 2}{|x-\alpha_j|}dx\\
-\int_{-2^L}^{\real(\alpha_j)}\frac{3\ln(\real(\alpha_j)-x+|\alpha_j-x|)}{|\alpha_j-x|}dx-\int_{\real(\alpha_j)}^{2^L}\frac{3\ln(x-\real(\alpha_j)+|\alpha_j-x|)}{|\alpha_j-x|}dx.\tag{\theequation a,\theequation b,\theequation c}
\end{multline*}
\end{subequations}
Integral \ref{equation:sqfreeeval:complex} is one of the integrals computed in \cite{BurrKrahmer:SqFreeEVAL} (up to a constant), and, there, it was shown to be $\widetilde{O}(dL)$.  In the final computation, this is a lower-order term, and it can be ignored.  Integrals \ref{equation:sqfreeeval:complex}a and \ref{equation:sqfreeeval:complex}b can be computed as follows:
$$
3(\ln(|\imag(\alpha_j)|))^2-\frac{3}{2}\left(\ln\left(2^L-\real(\alpha_j)+|\alpha_j-2^L|\right)\right)^2-
\frac{3}{2}\left(\ln\left(\real(\alpha_j)+2^L+|\alpha_j+2^L|\right)\right)^2.
$$
The terms with negative coefficient are bounded above by 0, and they can be ignored.  The terms with positive coefficient are the squares of the logarithms of half the distance between complex conjugates.  In \cite{BurrKrahmer:SqFreeEVAL}, these logarithms were shown to be $O(d(L+\ln d))$.  The squares are, therefore, $O(d^2(L+\ln d)^2)$.  Substituting everything in, we have that the bit complexity for $\SqFreeEVAL$ is $\widetilde{O}(d^4L^2)$, matching the bound in \cite{SagraloffYap:8Point,SagraloffYap:CEVAL}.

\begin{remark}\label{remark:sqfreeevalbit}
The bound above on the bit-complexity was first computed in \cite{SagraloffYap:8Point}; it could also be computed by calculating the maximum bit-complexity of any node in the tree and then multiplying this by the size of the tree, cf. \cite{SagraloffYap:CEVAL,BurrKrahmer:SqFreeEVAL}.  The adaptivity of different methods to compute the bit-complexity varies greatly; we will see that the continuous amortization calculation is one of the most sensitive calculations.

It is hard to determine the sensitivity of the calculation in \cite{SagraloffYap:8Point,SagraloffYap:CEVAL} because the bounds on the paper do not explicitly depend on the distances between roots.  In particular, the Mahler-Davenport bound is used on clusters of roots.  In addition, the distances between roots in clusters are bounded above and below and those bounds are used instead of the explicit distances between roots.  It would be a significant undertaking (if it is even possible) to reformulate the bounds in \cite{SagraloffYap:8Point,SagraloffYap:CEVAL} to be as adaptive as those presented here.

The computation of the bit-complexity of $\SqFreeEVAL$ can be performed in several ways.  In particular, we have different ways to calculate the size of the subdivision tree and the depth of the tree (e.g., the maximum depth of the tree can be used to bound the maximum bit-complexity of a node).  The maximum depth of the tree can be calculated as follows:
\begin{inparaenum}[(1)]
\item The simplest bound on the depth of the tree is to use the total size of the tree to bound the depth.
\item Using stopping functions, one can see that the depth of the tree is bounded above by $\widetilde{O}(L-\log_2(|\alpha_i-\alpha_j|)$ where $\alpha_i$ and $\alpha_j$ are the closest pair of roots of $p$ and $p'$.
\end{inparaenum}

In addition, the size of the tree can be calculated in several ways including
\begin{inparaenum}[(1)]
\item The complexity bound of \cite{SagraloffYap:8Point,SagraloffYap:CEVAL}.
\item The continuous amortization-based complexity bound in \cite{BurrKrahmer:SqFreeEVAL}.
\end{inparaenum}
Any combination of these approaches can be used, but one of the more adaptive ones is to use stopping functions and continuous amortization.

Ignoring constants and logarithmic terms, the bit-complexity using the size of the tree in \cite{BurrKrahmer:SqFreeEVAL} and a stopping function for the depth of the tree results in a bit-complexity of 
\begin{equation}\label{eq:SqFreeEVALBit}
d^2\cdot\left(dL-\sum \ln(|\alpha_{i_1}-\alpha_{i_2}|)\right)\cdot\max\{-\ln(|\alpha_{i_1}-\alpha_{i_2}|)\}
\end{equation} where $(\alpha_{i_1},\alpha_{i_2})$ are pairs of roots (assuming enough roots are sufficiently close).  The computation above, using the continuous amortization for bit complexity results in a bit complexity of 
\begin{equation}\label{eq:CAbiteval}
d^2\cdot\sum(\ln(|\alpha_{i_1}-\alpha_{i_2}|))^2
\end{equation}
where the pairs of $(\alpha_{i_1},\alpha_{i_2})$ are the same as above.  The most significant difference between these sums is that in Expression \ref{eq:SqFreeEVALBit}, the final factor is used for all root distances while in Expression \ref{eq:CAbiteval}, each pair of root distances is squared individually.

It is easiest to see the difference between these calculations with an example.
\begin{example}
Consider a polynomial where $\lambda d$ terms of the form $\ln(|\alpha_{i_1}-\alpha_{i_2}|)$ appear in the sums of Expressions \ref{eq:SqFreeEVALBit} and \ref{eq:CAbiteval} (where $0<\lambda\leq1$ is some positive constant).  Further, we assume that exactly one pair of roots has inter-root distance $O\left(e^{-\sqrt{d}L}\right)$ while the remaining $\lambda d-1$ pairs have inter-root distance of $O\left(e^{-L}\right)$.  In this case, the size of the tree may be $\widetilde{O}(dL)$ and the maximum depth will be $\widetilde{O}(\sqrt{d}L)$.  Note that these choices are reasonable because they do not violate the Mahler-Davenport bounds.  Evaluating the bit-complexity using the size of the subdivision tree as the maximum depth of the tree results in a total bit-complexity of $\widetilde{O}(d^4L^2)$.  Evaluating Expression \ref{eq:SqFreeEVALBit} results in a total bit-complexity of $\widetilde{O}\left(d^{7/2}L^2\right)$ while Expression \ref{eq:CAbiteval} results in a total bit-complexity of $\widetilde{O}(d^3L^2)$. Hence, in this case, the continuous amortization computation above is the most adaptive computation.
\end{example}
\end{remark}

\section{Conclusion and Discussion}
In this paper, we have extended the technique of continuous amortization by applying it to several different root isolation algorithms.  In addition, we have extended the technique to be able to compute complexities in higher dimensions as well as the bit-complexity algorithms.  This provides a unifying framework for the analysis of subdivision algorithms, and connects many of the previous computations in the literature.

One of the major goals of continuous amortization is to compute the complexity of two-dimensional algorithms for approximating planar curves, e.g., \cite{plantinga-vegter:isotopic:04,plantinga:thesis:06}.  The current paper is a step in this direction because we have extended the reach of continuous amortization and provided a useful collection of examples which aid in the understanding of how to apply continuous amortization to new problems.

There are many algorithms for which continuous amortization is applicable, but have not yet been studied.  Some reasonable problems for the next step in this program is to apply continuous amortization to continued fractions \cite{vincentcontinuedfractions} and the non-8-point test of $\SqFreeCEVAL$ \cite{SagraloffYap:CEVAL}.  As was described in Section \ref{section:continuousamortization}, the main challenge with continued fraction techniques is that the bisections are not uniform in size.  On the other hand, the main challenge with the full version of \texttt{CEVAL} in \cite{SagraloffYap:CEVAL}, is that, in the notation of \cite{SagraloffYap:CEVAL}, the test $T'_{\sqrt{2}}(m,4dr)$ includes a $d$ on the right-hand-side of the inequality.  If the stopping functions of this paper are adapted na\"ively to the inequality $T'_{\sqrt{2}}(m,4dr)$, then the corresponding bound from continuous amortization increases by a factor of at least $O(d^2)$.  To be able to apply continuous amortization to achieve state-of-the-art complexity bounds for this algorithm, a new stopping function must be developed.

\bibliographystyle{plain}
\bibliography{ContinuousAmortization}

\end{document}